\documentclass[a4paper]{article}

\usepackage[margin=1in]{geometry}  
\usepackage[english]{babel}
\usepackage[utf8x]{inputenc}
\usepackage[T1]{fontenc}

\usepackage[dvipsnames]{xcolor}
\usepackage{graphicx}
\usepackage{longtable} 
\usepackage{multirow}
\usepackage{listings}
\usepackage{makecell}
\usepackage{array}
\usepackage{float}
\usepackage{dsfont}
\usepackage{rotating}
\usepackage{booktabs}
\usepackage{enumitem}
\usepackage{tikz}
\usepackage{pgf}

\usepackage{amsmath}
\usepackage{amssymb}
\usepackage{amsthm}
\usepackage{bm}
\usepackage[ruled]{algorithm2e}
\usepackage{mathtools}
\usepackage{mathrsfs}
\usepackage{subfigure}
\usepackage{comment}

\newcommand{\bbR}{\mathbb{R}}

\usepackage{natbib}
\usepackage[
  colorlinks,
  citecolor=blue,
  linkcolor=red,
  anchorcolor=red,
  urlcolor=blue
]{hyperref}
\usepackage{authblk}
\usepackage{todonotes}
\theoremstyle{plain}

\newtheorem{theorem}{Theorem}

\newtheorem{assumption}{Assumption}

\newcommand{\var}{\mathrm{Var}}

\usepackage{xspace}

\newcommand{\calZ}{{\mathcal{Z}}}

\newcommand{\bbE}{\mathbb{E}}
\newcommand{\I}{{\mathbf{I}}}

\newcommand{\htheta}{\widehat{\theta}}
\newcommand{\ttheta}{\widetilde{\theta}}
\newcommand{\hthetai}{\widehat{\theta}_{-i}}

\newcommand{\prox}{{\rm prox}}
\newcommand{\CV}{{\rm CV}}
\newcommand{\NS}{{\rm NS}}
\renewcommand{\IJ}{{\rm IJ}}

\begin{document}
\title{Iterative Approximate Cross-Validation}
\author{Yuetian Luo$^1$, ~ Zhimei Ren$^2$, ~ and ~ Rina Foygel Barber$^2$}
\date{}
\footnotetext[1]{Data Science Institute, University of Chicago}
\footnotetext[2]{Department of Statistics, University of Chicago}

\maketitle

\begin{abstract}
	Cross-validation (CV) is one of the most popular tools for assessing and selecting predictive models. However, standard CV suffers from high computational cost when the number of folds is large. Recently, under the empirical risk minimization (ERM) framework, a line of works proposed efficient methods to approximate CV based on the solution of the ERM problem trained on the full dataset. 
	However, in large-scale problems, it can be hard to obtain the exact solution of the ERM problem, either due to limited computational resources or due to early stopping as a way of preventing overfitting. 
		In this paper, we propose a new paradigm to efficiently approximate CV when the ERM problem is solved via an iterative first-order algorithm, without running until convergence. Our new method extends existing guarantees for CV approximation to hold along the whole trajectory of the algorithm, including at convergence, thus generalizing
		existing  CV approximation methods. Finally, we illustrate the accuracy and computational efficiency of our method through a range of empirical studies.
\end{abstract}

\section{Introduction} \label{sec: intro}
In machine learning and statistics, cross-validation (CV) \citep{allen1974relationship,stone1974cross,geisser1975predictive} is one of the most popular methods for the tasks of assessing and selecting predictive models. It is conceptually simple and easy to implement. Among many variants of CV, one popular choice is leave-one-out CV (also
called the Jackknife), which often offers the most accurate prediction for the out-of-sample risk in many practical and challenging scenarios \citep{arlot2010survey,stephenson2020approximate,rad2020scalable}. Leave-one-out CV uses $n-1$ out of $n$ data points for training, and the remaining one for testing,
and then repeats for each of the $n$ data points in the sample. The resulting $n$ cross-validation errors can be used for model assessment or selection, tuning parameter selection, etc. 

However, the superior performance of leave-one-out CV is accompanied by high computational cost when $n$ is large, as the potentially complex model needs to be fitted $n$ times. To be able to run leave-one-out CV but avoid the high computational burden, several methods have been proposed to approximate the expensive model refitting step with an inexpensive surrogate step, with much progress in the recent literature particularly in the context of empirical risk minimization (ERM)~\citep{obuchi2016cross,beirami2017optimal,giordano2019swiss,koh2017understanding,wang2018approximate,wilson2020approximate,rad2020scalable,stephenson2020approximate}.

\paragraph{Background: approximate CV for ERM} 
Suppose one aims to estimate some parameter of interest $\theta^*$ via solving the following regularized ERM problem:
\begin{equation} \label{eq: main-objective}
	\htheta = \arg\!\min_{\theta \in \bbR^p} F(\calZ;\theta),
\end{equation} 
where, for a dataset $\calZ$ of size $n$,
\[F(\calZ;\theta):=  \sum_{i=1}^n \ell(Z_i; \theta) + \lambda \pi(\theta).\]
Here $\ell(Z; \theta)$ is the loss function for data point $Z$ and parameter $\theta$, $\calZ = \{Z_i\}_{i=1}^n$ is the set of observed data points, $\pi(\cdot)$ is a regularization term, and $\lambda \geq 0$ is a tuning parameter. Under this setting, the widely used leave-one-out CV loss for estimating the prediction performance of $\htheta$ is given by 
\begin{equation} \label{eq: canonical-exact-cv}
	\CV(\{\hthetai \}_{i=1}^n ) = \frac{1}{n} \sum_{i=1}^n \ell(Z_i; \hthetai),
\end{equation} 
where $\hthetai := \arg\!\min_{\theta \in \bbR^p} F(\calZ_{-i};\theta)$ is the minimizer of the leave-one-out objective $F(\calZ_{-i};\theta) :=   \sum_{j=1, j\neq i}^n \ell(Z_j; \theta) + \lambda \pi(\theta)$. 
Since computing $\hthetai $ for every $i\in [n]:=\{1,\ldots,n\}$ is often expensive, this motivates finding an approximation
to $\CV(\{\hthetai \}_{i=1}^n ) $ that does not require fitting $n$ many models.

Observe that $\htheta$ and $\hthetai$ are the solutions to two highly similar minimization problems---they minimize
$F(\calZ;\theta)$  and $F(\calZ_{-i};\theta) = F(\calZ;\theta) - \ell(Z_i;\theta)$, respectively.
One approach in the recent literature
 proposes approximating $\hthetai$ by initializing at $\theta=\htheta$ and taking a single Newton step (NS) on the objective $F(\calZ_{-i};\theta)$ \citep{beirami2017optimal,rad2020scalable},
\begin{equation}\label{eqn:NS-intro}
	\ttheta^{\NS}_{-i} = \htheta - \left(\nabla^2_{\theta} F ( \calZ_{-i}; \htheta ) \right)^{-1} \nabla_{\theta} F ( \calZ_{-i}; \htheta ).
\end{equation}
Another approach is based on the infinitesimal jackknife (IJ) \citep{jaeckel1972infinitesimal,efron1982jackknife,giordano2019swiss}:
 \begin{equation}\label{eqn:IJ-intro}
	\ttheta^{\IJ}_{-i} = \htheta - \left(\nabla^2_{\theta} F ( \calZ; \htheta ) \right)^{-1} \nabla_{\theta} F ( \calZ_{-i}; \htheta ).
\end{equation}
(Here for simplicity we consider the case where $F$ is twice differentiable; we will discuss  other settings below.)

\paragraph{The challenge: iterative algorithms}

The accuracy of the estimators in \eqref{eqn:NS-intro} and \eqref{eqn:IJ-intro} relies heavily on the assumption that $\htheta$ can be computed exactly, which may not be the case in many scenarios. For example, when \eqref{eq: main-objective} is solved via an iterative algorithm such as gradient descent (GD), we may not be able to run the algorithm to convergence due to the limited computational resources. Other algorithms we might use have a very slow rate of convergence, such as stochastic gradient descent (SGD). Another important example is that, in some settings, we may intentionally stop training early to avoid overfitting in learning machine learning models (see e.g. \cite{jabbar2015methods})

Figure~\ref{fig: intro} shows a simple simulation that illustrates the loss of accuracy suffered by the NS and IJ methods,
 when run with an estimate of $\htheta$ obtained before convergence. Specifically, consider logistic regression with a ridge penalty, with the estimator obtained by running GD or SGD for $t$ steps---that is, we would like to estimate $\CV(\{\htheta^{(t)}_{-i}\}_{i=1}^n)$, where $\htheta^{(t)}_{-i}$ is the solution after $t$ steps of GD or SGD on the objective function $F(\calZ_{-i};\theta)$.
 $\htheta^{(t)}_{-i}$  is then approximated by running either NS~\eqref{eqn:NS-intro} or IJ~\eqref{eqn:IJ-intro}, but
with $\htheta^{(t)}$ (the $t$th step of GD or SGD on the full objective function $F(\calZ;\theta)$) in place of the exact minimizer $\htheta$. We see highly inaccurate approximations of NS and IJ methods to the leave-one-out CV loss $\CV(\{\hthetai^{(t)} \}_{i=1}^n )$ during the early iterations of GD, and even after 1000 iterations for SGD. 
Details of this simulation are given in Section \ref{sec: simulation}.

\paragraph{Our contribution: iterative approximate CV} 
The central question of our paper is this:
\begin{quote}
 Can we develop new approximation schemes to  leave-one-out CV when $\htheta$ is not known exactly, but is estimated with an iterative algorithm that is not run to convergence?\end{quote}
In this work, we provide a positive answer to this question. We find that before the convergence of the iterative algorithm, efficient leave-one-out CV approximation can still be possible if we can leverage information about the iterative nature of the algorithm. In Figure~\ref{fig: intro}, we see that our proposed method, Iterative Approximate CV (IACV), achieves an accurate approximation to  leave-one-out CV loss $\CV(\{\hthetai^{(t)} \}_{i=1}^n )$ across all iterations $t$, even far before convergence; at convergence, IACV yields the same results as NS. Theoretically, we are able to show that, under some regularity conditions, IACV enjoys guaranteed CV approximation along the whole trajectory of the algorithm, and moreover, at convergence, IACV recovers the NS method~\eqref{eqn:NS-intro}.
\begin{figure}[t]
	\centering\vspace{0.2in}
	\includegraphics[width = 0.7\textwidth]{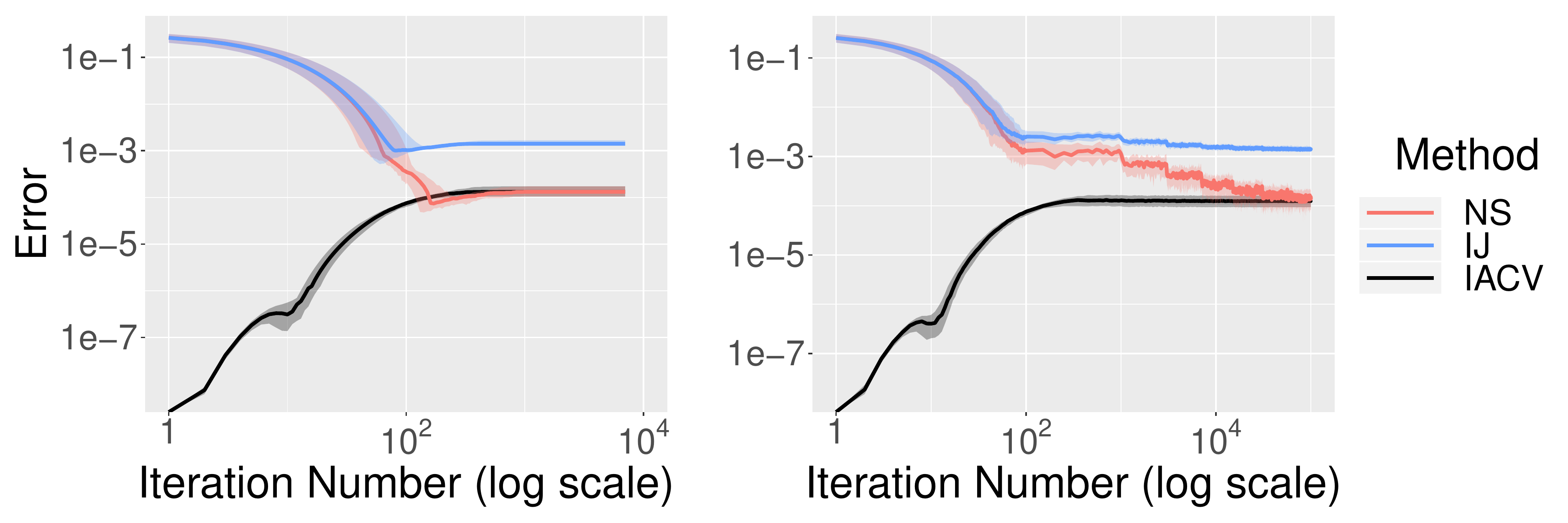}
	\caption{The red and blue lines show the error in estimating $\CV(\{\htheta^{(t)}_{-i}\}_{i=1}^n)$  for the NS and IJ methods, run with $\htheta^{(t)}$ as an approximation to the exact solution $\htheta$, while the black line shows the error for our proposed method, iterative approximate CV (IACV). The objective function is given by logistic regression plus a ridge penalty, solved iteratively with GD (left, $n = 1000$) and SGD (right, $n = 1000$, batch-size = $400$). Solid lines represent the median value over $100$ repetitions and the shaded area shows the region between lower and upper quartiles.} \label{fig: intro}
\end{figure}

\subsection{Additional Related Literature} \label{sec: related-literature}

In addition to the NS and IJ methods, described earlier, many other works in the literature also offer methods and theories for the problem of approximating  leave-one-out CV.
The closest to ours is the work by \citet{koh2017understanding,ghosh2020approximate}, which also considers the problem of approximating CV with an inexact $\htheta$, but with results of a very different flavor---this line of work assumes that $\htheta^{(t)}$ is already an accurate estimate of $\htheta$ (i.e., $t$ is large and the iterative algorithm is near convergence), and bounds the error in approximating the leave-one-out models $\htheta_{-i}$ as a function of this convergence error $\|\htheta^{(t)}-\htheta\|_2$. In contrast, our work instead estimates the models $\htheta^{(t)}_{-i}$ (i.e., at each time $t$ rather than at convergence), and does not assume that $\|\htheta^{(t)}-\htheta\|_2$ is small.

Another recent line of work studies the problem of efficient data deletion in a trained model under the ERM framework, where the aim is to provide an estimator of $\theta$ that is approximately independent of data point $i$, for the sake of the $i$th individual's privacy. These problems are often referred to as ``data deletion'', ``machine unlearning'', or ``decremental learning'' in the literature \citep{tsai2014incremental,cao2015towards,bourtoule2021machine,neel2021descent}. Like our work, these methods also aim to approximate the leave-one-out model efficiently given the trained full model, but additionally require that the approximate model is ``statistically indistinguishable'' from a model that would have resulted from retraining without the $i$th individual's data \citep{izzo2021approximate,guo2020certified,ginart2019making}. These works generally need to add noise in the approximation procedure in order to ensure this constraint is satisfied. 

\section{Iterative Approximate CV} \label{sec: methodology}
In this section, we present our method for approximating  leave-one-out CV in the setting of iterative optimization, where the algorithm may not be run to convergence.
We find that even when $\htheta$ can not be exactly computed, we can still approximate the leave-one-out iterates $\htheta^{(t)}_{-i}$ accurately by leveraging the structure of the iterative update. In our theoretical results below, we will see that our estimates have guaranteed accuracy (under mild conditions) along all iterations $t\geq 1$ of the algorithm, i.e., even at iterations when the algorithm is far from convergence.

\subsection{Framework}
The iterative procedures that we will study in this paper---gradient descent, stochastic gradient descent, and proximal gradient descent---can all be expressed within the following general framework.
Suppose our objective function can be written in the form
\[F(\calZ; \theta) = g(\calZ; \theta) + h(\theta)\]
where $g(\calZ; \theta)$ is twice-differentiable in $\theta$ while $h(\theta)$ may be nondifferentiable. We consider solving the problem \eqref{eq: main-objective} iteratively as follows: for each step $t \geq 1$, we take a gradient step on $g$ (possibly after subsampling the data points), and a proximal step on $h$. Specifically, the iterations are given by
\begin{equation}\label{eqn: iter-general}\htheta^{(t)} = 
\arg\!\min_\theta\left\{ \frac{1}{2\alpha_t} \|\theta - \theta'\|^2_2 + h(\theta)\right\} \quad 
\textnormal{ where }\quad \theta' = \htheta^{(t-1)} - \alpha_t\nabla_\theta g(\calZ_{S_t};\htheta^{(t-1)}),\end{equation}
where $S_t\subseteq[n]$ is a subset of indices, $\calZ_{S_t}:= \{Z_i: i \in S_t \}$ is the corresponding subset of the data, and $\alpha_t>0$ is the learning rate. 

In the setting where the objective function is twice-differentiable, we can take $g(\calZ;\theta) = F(\calZ;\theta)$ and $h\equiv 0$. For example, in the regularized ERM setting, if the loss function $\ell$ and the penalty function $\pi(\theta)$ are twice-differentiable, we can simply take $g(\calZ;\theta) = \sum_{i=1}^n \ell(Z_i;\theta) + \lambda\pi(\theta)$. Then~\eqref{eqn: iter-general} yields gradient descent if we choose the full dataset $S_t = [n]$ at each iteration,
\begin{equation} \label{eq: gd-iterative-scheme}
	\textnormal{GD}: \quad \htheta^{(t)} =  \htheta^{(t-1)} - \alpha_t \nabla_{\theta} F(\calZ; \htheta^{(t-1)}  ),
\end{equation}
or stochastic gradient descent if at each iteration we sample a random batch $S_t\subseteq[n]$,
\begin{equation} \label{eq: sgd-iterative-scheme}
	\textnormal{SGD}: \quad \htheta^{(t)} =  \htheta^{(t-1)} - \alpha_t \nabla_{\theta} F(\calZ_{S_t}; \htheta^{(t-1)}  ).
\end{equation}

In other settings, we may have a nondifferentiable term $h(\theta)$,
which has an inexpensive proximal map (i.e., the solution to~\eqref{eqn: iter-general} can be computed efficiently for each iteration). For instance, in regularized ERM with a nonsmooth penalty function $\pi$ such as the $\ell_1$ norm, we might take $g(\calZ;\theta) = \sum_{i=1}^n \ell(Z_i;\theta)$ and $h(\theta) = \lambda\pi(\theta)$. Then the general iterative scheme~\eqref{eqn: iter-general}, run again with the full dataset $S_t=[n]$ at each iteration, reduces to proximal gradient descent,
\begin{equation}\label{eq: proxgd-iterative-scheme}
		\textnormal{ProxGD}: \, 
  \htheta^{(t)} = \arg\!\min_\theta\Big\{ \frac{1}{2\alpha_t} \|\theta - \theta'\|^2_2 + h(\theta)\Big\} \quad
\textnormal{ where }\quad \theta' = \htheta^{(t-1)} - \alpha_t\nabla_\theta g(\calZ;\htheta^{(t-1)}).
\end{equation}

\subsection{Estimation Procedure: General Case}
We will now provide an algorithm for IACV for general iterative procedures of the form~\eqref{eqn: iter-general}. After deriving the general formulation, we will then show how it specializes to each of the three settings, GD, SGD, and ProxGD.

The targets $\htheta^{(t)}_{-i}$ for $i\in[n]$ are obtained by running the same iterative solver as in \eqref{eqn: iter-general} except with $i$-th data point being left out: for each $t\geq 1$,
\begin{equation}\label{eqn: iter-general-minus-i}
	\htheta^{(t)}_{-i} = \arg\!\min_\theta\left\{ \frac{1}{2\alpha_t} \|\theta - \theta'\|^2_2 + h(\theta)\right\} \quad 
\textnormal{ where }\theta'= \htheta^{(t-1)}_{-i} - \alpha_t\nabla_\theta g(\calZ_{S_t\setminus\{i\}};\htheta^{(t-1)}_{-i}).
\end{equation}

In many settings, running the iterative algorithm for each $i\in[n]$ is computationally too expensive, since at each iteration $t$, it requires computing gradients for $g$ at $n$ different parameter vectors $\{\htheta^{(t-1)}_{-i}\}_{i=1}^n$.
Thus, we aim to find computationally inexpensive surrogates for these gradients.

We now define our procedure, which produces approximations $\ttheta^{(t)}_{-i}\approx \htheta^{(t)}_{-i}$, at each iteration $t\geq 1$ and for each $i\in[n]$.
At iteration $t$, if we are not running the exact leave-one-out procedure, then computing $\theta'$ in~\eqref{eqn: iter-general-minus-i} above involves two unknowns: the previous iterate, $\htheta^{(t-1)}_{-i}$, which we simply approximate with our previous iteration's estimate $\ttheta^{(t-1)}_{-i}$, and its gradient, $\nabla_\theta g(\calZ_{S_t\setminus\{i\}};\htheta^{(t-1)}_{-i})$, which we now address.
If $\theta\mapsto g(\cdot; \theta)$ is twice-differentiable, a natural idea is to approximate $\nabla_\theta g(\calZ_{ S_t \setminus i}; \htheta^{(t-1)}_{-i})$ by its Taylor expansion at $\htheta^{(t-1)}$. The reason we choose $\htheta^{(t-1)}$ as the base point is that it can be shared across the $n$ problems of approximating $\nabla_\theta g(\calZ_{ S_t \setminus i}; \htheta^{(t-1)}_{-i})$ for each $i \in [n]$. We thus take the approximation
\begin{equation}\label{eq: main-taylor-expansion-general}
    	\nabla_\theta g(\calZ_{ S_t \setminus i}; \htheta^{(t-1)}_{-i}) \approx \nabla_\theta g(\calZ_{ S_t \setminus i}; \htheta^{(t-1)}) 
	+ \nabla^2_\theta g(\calZ_{S_t \setminus i}; \htheta^{(t-1)})[ \hthetai^{(t-1)} - \htheta^{(t-1)} ],
\end{equation}
and then again plug in $\ttheta^{(t-1)}_{-i}$ for the unknown $\htheta^{(t-1)}_{-i}$ in the last term.

Thus, for each $i\in[n]$ each $t\geq 1$, the IACV approximation is given by
\begin{equation}\label{eqn: iter-IACV}
	\begin{split}
		\textnormal{IACV:} \ \ &\ttheta^{(t)}_{-i} = \arg\!\min_\theta\left\{ \frac{1}{2\alpha_t} \|\theta - \theta'\|^2_2 + h(\theta)\right\}\\
& \textnormal{ where }   \theta'= \ttheta^{(t-1)}_{-i} - \alpha_t\big(\nabla_\theta g(\calZ_{ S_t \setminus i}; \htheta^{(t-1)}) 
	+ \nabla^2_\theta g(\calZ_{S_t \setminus i}; \htheta^{(t-1)})[ \ttheta_{-i}^{(t-1)} - \htheta^{(t-1)} ]\big).
	\end{split}
\end{equation}

\subsection{Estimation Procedure for GD, SGD, and ProxGD}
Next, we give the details for implementing this general procedure in the specific settings of GD, SGD, and ProxGD.

\paragraph{Gradient descent (GD)}
For GD, we have $h(\theta)\equiv 0$ and $S_t \equiv [n]$, so the steps of IACV reduce to
\begin{equation}\label{eqn: iter-IACV-GD}
	\ttheta^{(t)}_{-i}= \ttheta^{(t-1)}_{-i} - \alpha_t\big(\nabla_\theta F(\calZ_{-i}; \htheta^{(t-1)}) 
	+ \nabla^2_\theta F(\calZ_{-i}; \htheta^{(t-1)})[ \ttheta_{-i}^{(t-1)} - \htheta^{(t-1)} ]\big).
\end{equation}
 
\paragraph{Stochastic gradient descent (SGD)}
For SGD, we again have $h(\theta)\equiv 0$, but the sets $S_t$ consist of small batches of data. The steps of IACV reduce to
\begin{equation}\label{eqn: iter-IACV-SGD}
    \ttheta^{(t)}_{-i}= \ttheta^{(t-1)}_{-i} - \alpha_t\big(\nabla_\theta F(\calZ_{S_t\setminus i}; \htheta^{(t-1)}) 
	+ \nabla^2_\theta F(\calZ_{S_t\setminus i}; \htheta^{(t-1)})[ \ttheta_{-i}^{(t-1)} - \htheta^{(t-1)} ]\big).
\end{equation}

\paragraph{Proximal gradient descent (ProxGD)}
For ProxGD, we again take $S_t\equiv [n]$, but the function $h(\theta)$ is nontrivial (e.g., a nonsmooth regularizer). The steps of IACV reduce to
\begin{equation}\label{eqn: iter-IACV-ProxGD}
	\begin{split}
		&\ttheta^{(t)}_{-i}=\arg\!\min_\theta\left\{\frac{1}{2\alpha_t}\|\theta - \theta'\|^2_2 + h(\theta)\right\}\\
\textnormal{ where } &\theta' = \ttheta^{(t-1)}_{-i} - \alpha_t\big(\nabla_\theta g(\calZ_{-i}; \htheta^{(t-1)}) 
	+ \nabla^2_\theta g(\calZ_{-i}; \htheta^{(t-1)})[ \ttheta_{-i}^{(t-1)} - \htheta^{(t-1)} ]\big).
	\end{split}
\end{equation}

\section{Computation and Memory Complexity} \label{sec: comp-memory-complexity-computation}
In this section, we compare the computation and memory complexity of the proposed leave-one-out CV approximation method~\eqref{eqn: iter-IACV} with the exact leave-one-out CV~\eqref{eqn: iter-general-minus-i}.
Note that running IACV~\eqref{eqn: iter-IACV} assumes that we have access to the iterates $\htheta^{(t)}$, obtained by running the iterative procedure~\eqref{eqn: iter-general} on the full dataset. In addition, at each time $t$, we also need to calculate $\nabla_\theta g(\calZ_{ S_t \setminus i}; \htheta^{(t-1)})$ and $ \nabla^2_\theta g(\calZ_{S_t \setminus i}; \htheta^{(t-1)})$ in order to compute the estimator. In other words,
at each time $t$ we need to compute the gradient and Hessian of $g$ for datasets $\calZ_{S_t\setminus i}$ that {\em vary} with $i$ but at a parameter vector $\htheta^{(t-1)}$ that is {\em constant} over $i$---in general, this can be done efficiently due to the separability of the data in the ERM problem. This is different from the computational cost of exact leave-one-out CV, which instead requires computing $\nabla_\theta g(\calZ_{ S_t \setminus i}; \htheta^{(t-1)}_{-i})$ for each $i$ at each iteration $t$, i.e., the gradient of $g$ needs to be computed at $n$ {\em different} parameter vectors $\{\htheta^{(t-1)}_{-i}\}_{i=1}^n$.

For simplicity, let us consider the case where $g(\calZ;\theta) = \sum_{i=1}^n \ell(Z_i;\theta)$. We denote the computational cost of one call to $\nabla_\theta \ell(Z_i, \cdot)$ as $A_p$, and one call to $\nabla^2_\theta \ell(Z_i, \cdot)$ as $B_p$, which depend on the dimension $p$ of the parameter $\theta$. We also assume that the cost of computing $\nabla_\theta g(\calZ;\cdot) = \sum_{i=1}^n \nabla_\theta \ell(Z_i;\cdot)$ is equal to $nA_p$, and similarly $nB_p$ for $\nabla^2_\theta g(\calZ;\cdot)$; this assumption is mild since, in most typical settings, the gradient or Hessian of the loss can only be computed via evaluation on each data point one-by-one.
Finally, let $D_p$ denotes the computational cost of one call to the proximal operator $\theta' \mapsto\arg\!\min_\theta\{\frac{1}{2\alpha_t}\|\theta - \theta'\|^2_2 + h(\theta)\}$.

\begin{table}[t]
	\centering
	\begin{tabular}{ @{\,}c@{\,}| @{\,}c@{\,} | @{\,}c@{\,} }
		& \multicolumn{2}{c}{Computational Cost} \\
		 \hline
		& IACV & Exact LOO CV  \\
		 \hline
	   GD &	 $n(A_p+B_p) + np^2$ & $n^2 A_p + np$  \\
		 \hline
            SGD  & $K(A_p+B_p) + np^2$ & $n K A_p + np$ \\
 		 \hline  
		ProxGD	& $n(A_p+B_p +D_p) + np^2 $ & $n^2 A_p + nD_p + np$ \\
		 \hline
	\end{tabular}
	\caption{The order of per-iteration computation complexity of iterative approximate CV (IACV) as compared to exact leave-one-out (LOO) CV. See Section~\ref{sec: comp-memory-complexity-computation} for details.} \vspace{-.1in}\label{tab: time-memory-complexity}
\end{table}

In Table \ref{tab: time-memory-complexity}, we provide the per-iteration computation complexity of computing $\{ \ttheta_{-i}^{(t)} \}_{i=1}^n$ for a single iteration $t$ in our proposed method IACV, as compared to computing $\{ \hthetai^{(t)} \}_{i=1}^n$ in the exact leave-one-out CV, for the three iterative algorithms GD, SGD (with batch size $K$), and ProxGD. 
Details for these calculations are given in Appendix~\ref{app:computational-complexity}.

For example, for GD, we can see that IACV is more efficient as long as $B_p + p^2\ll nA_p$.  Since typically, $A_p$ and $B_p$ are of order $p$ and $p^2$, respectively, we can see that condition $B_p + p^2 \ll nA_p$ is typically satisfied when $n \gg p$ (which is exactly the regime considered in the literature for analyzing the existing CV approximation methods \citep{beirami2017optimal,wilson2020approximate}), and thus IACV is computationally more efficient in this regime.

One limitation of IACV is an increased memory cost---the proposed method needs to store the data, gradient, and Hessian, which costs $O(np + p^2)$ space, while the exact leave-one-out CV costs $O(np)$ as it only needs to store the data and gradient. In the $n \gg p$ regime, however, 
the memory cost of the two methods are on the same order.

\section{Accuracy Guarantees for IACV} \label{sec: approximation-property}
 In this section, we provide theoretical guarantees for IACV in the GD and SGD settings. (Guarantees for ProxGD are given in Appendix~\ref{app: prox-GD-theory}.) For both GD and SGD we will consider the objective function $F(\calZ;\theta) = \sum_{i=1}^n\ell(Z_i;\theta)+\lambda\pi(\theta)$, where $\pi$ is a twice-differentiable regularizer, as before.
 
 We will study two notions of accuracy: the {\em approximation error}, \begin{equation}\label{eqn: define-approx-error}\textnormal{Err}_{\textnormal{approx}}^{(t)} = \frac{1}{n}\sum_{i=1}^n \|\ttheta^{(t)}_{-i} - \htheta^{(t)}_{-i}\|_2,\end{equation}
measuring our accuracy in approximating the exact leave-one-out estimators $\htheta^{(t)}_{-i}$, and the {\em CV error},
\begin{equation}\label{eqn: define-model-assessment-error}\textnormal{Err}_{\textnormal{CV}}^{(t)} = |\CV(\{\ttheta^{(t)}_{-i}\}_{i=1}^n) - 
\CV(\{\htheta^{(t)}_{-i}\}_{i=1}^n)|, 
\end{equation}
measuring our accuracy in estimating the leave-one-out CV loss for $\{\htheta^{(t)}_{-i} \}_{i = 1}^n$. (This latter notion of error is the quantity that was plotted in Figure~\ref{fig: intro}.)

\subsection{Guarantees for GD} \label{sec: guarantee-for-GD}
We first state our assumptions for the GD setting.
\begin{assumption}\label{asm: hessian-condition} 
 For all $ t \geq 1$, $$ \min_{i \in [n]} \lambda_{\min}( \nabla^2_\theta F (\calZ_{-i}; \htheta^{(t-1)} ) ) \geq n \lambda_0 \quad  \text{ and } \quad  \max_{i \in [n]}  \lambda_{\max}( \nabla^2_\theta F (\calZ_{-i}; \htheta^{(t-1)} ) )  \leq n\lambda_1$$
	for some positive constants $\lambda_0, \lambda_1 > 0$.
\end{assumption} 
\noindent This first assumption says the Hessian of $F$ is well-conditioned along all iterates. Note that in the literature for analyzing the NS~\eqref{eqn:NS-intro} or IJ~\eqref{eqn:IJ-intro} estimators, the Hessian is often assumed to be well-conditioned at $\htheta$, the minimizer of \eqref{eq: main-objective} \citep{beirami2017optimal,wilson2020approximate,giordano2019swiss}, which is sufficient since $\ttheta_{-i}^{\NS}$ and $\ttheta_{-i}^{\IJ}$ in \eqref{eqn:NS-intro} and \eqref{eqn:IJ-intro} are based on $\htheta$ only. Here we require a  stronger assumption because our estimator depends on $\htheta^{(t)}$ for all $t$.

The second assumption is about the gradient. 
\begin{assumption}\label{asm: gradient-control} 
	For all $ t \geq 1$, $i \in [n]$, we have $\| \nabla_\theta \ell(Z_i; \htheta^{(t-1)} ) \|_2\leq \eta_i $ for some $\eta_i > 0$. 
\end{assumption}

The final assumption requires the Hessian to be Lipschitz. 
\begin{assumption}\label{asm: third-derivative-control} 
For all $i \in [n]$, $ \nabla_\theta^2 F(\calZ_{-i}; \cdot)$ is $n\gamma$-Lipschitz, i.e., for any $\theta_1, \theta_2 \in \bbR^p$,
\begin{equation*}
    \| \nabla_\theta^2 F(\calZ_{-i}; \theta_1) -\nabla_\theta^2 F(\calZ_{-i}; \theta_2)  \| \leq n \gamma \|\theta_1 - \theta_2\|_2.
\end{equation*}
\end{assumption}

We are now ready to state our first main result about IACV. 
\begin{theorem}[Approximation Error of IACV for GD] \label{th: GD-case-approx-error}
	 Suppose $\htheta^{(0)} = \htheta^{(0)}_{-i} = \ttheta^{(0)}_{-i}$ for all $i \in [n]$, $\alpha_t \leq 1/(n \lambda_1)$ for $t \geq 1$, and Assumptions \ref{asm: hessian-condition}--\ref{asm: third-derivative-control} are satisfied with $\gamma < 2 \lambda_0$ and $n \geq \frac{4 \|\eta\|_\infty }{2 \lambda_0 - \gamma} $. Then for all $ t \geq 1$ and $i \in [n]$, we have \[\| \htheta^{(t)} - \htheta^{(t)}_{-i} \|_2\leq \frac{2 \eta_i}{(2 \lambda_0 - \gamma) n}\]
 and
 \[\|\ttheta_{-i}^{(t)} - \htheta^{(t)}_{-i}\|_2\leq \frac{4 \gamma \eta_i^2 }{\lambda_0 (2 \lambda_0 - \gamma)^2 n^2 }.\]
\end{theorem}
\noindent In Theorem \ref{th: GD-case-approx-error}, we give upper bounds that are independent of $t$, for clarity of the result; tighter, iteration-dependent upper bounds are discussed in Appendix \ref{app: a-time-dependent-bound}. As an immediate consequence of the last bound, the approximation error~\eqref{eqn: define-approx-error} is bounded as \[\textnormal{Err}_{\textnormal{approx}}^{(t)}\leq 
\frac{4 \gamma \|\eta\|^2_\infty}{\lambda_0 (2 \lambda_0 - \gamma)^2 n^2 }\textnormal{ for all $t\geq 1$.}\]

\vskip.1cm

	In Theorem \ref{th: GD-case-approx-error}, we see that the error bound for $ \|\ttheta_{-i}^{(t)} - \htheta^{(t)}_{-i}\|_2$ is of a smaller order than the one for $\| \htheta^{(t)} - \htheta^{(t)}_{-i} \|_2$. This is exactly what we need---the goal of leave-one-out CV is to determine how the leave-one-out estimators $\htheta^{(t)}_{-i}$ behave differently from the estimator $\htheta^{(t)}$ computed on the full dataset, and thus the approximations $\ttheta^{(t)}_{-i}$ are useful only if they can improve over the ``baseline'' accuracy of $\htheta^{(t)}$ itself.
 
 In addition, it has been shown in \citet{beirami2017optimal,wilson2020approximate} that the NS~\eqref{eqn:NS-intro} and IJ~\eqref{eqn:IJ-intro} estimators achieve an approximation error of order $1/n^2$, which is the same as our result above. However, NS and IJ only achieve this error when run with the exact value of $\htheta$ (as we have seen in Figure \ref{fig: intro}, this assumption is crucial for empirical accuracy), while IACV achieves the $O(1/n^2)$ approximation error bound along all iterations of the algorithm rather than only at convergence.

\vskip.1cm
Based on the approximation error bounds established in Theorem \ref{th: GD-case-approx-error}, we are also able to provide guarantees for IACV's ability to approximate the leave-one-out CV loss $\CV(\{\htheta^{(t)}_{-i}\}_{i=1}^n)$, with error on the order of $1/n^2$. We will need one additional assumption on the gradient of the loss, which can be viewed as a stronger version of Assumption~\ref{asm: gradient-control}:
\begin{assumption}\label{asm: GD-CV-error}
 Suppose for any $i \in [n]$ and all $t \geq 1$, 
    there exists $\eta'_i > 0$ such that \[\sup_{a\in[0,1]} \|\nabla \ell(Z_i; a \ttheta^{(t)}_{-i} + (1-a) \htheta^{(t)}_{-i} ) \|_2 \leq \eta_i'.\]
\end{assumption}
\noindent \noindent For example, this assumption will be satisfied if $\ell(Z_i,\cdot)$ is $\eta'_i$-Lipschitz for each $i$.
\begin{theorem}[CV Error of IACV for GD] \label{th: model-assessment-bound}
If the assumptions in Theorem \ref{th: GD-case-approx-error} are satisfied, and Assumption~\ref{asm: GD-CV-error} holds, then for all $t \geq 1$:
   \begin{equation*}
 \begin{split}
  \textnormal{Err}_{\textnormal{CV}}^{(t)} \leq \frac{1}{n} \sum_{i=1}^n \frac{4 \gamma \eta'_i \eta_i^2}{\lambda_0 (2 \lambda_0 - \gamma)^2 n^2 }. 
 \end{split}
 \end{equation*}
\end{theorem}
Theorem \ref{th: model-assessment-bound} suggests IACV might be useful when we perform early stopping based on CV loss as it has guaranteed accurate estimates of $\CV(\{\widehat{\theta}^{(t)}_{-i}\}_{i=1}^n)$ along the entire trajectory, including times $t$ that are well before convergence.

\paragraph{When are these assumptions satisfied?}
	Next, we provide an example of a setting where Assumptions \ref{asm: hessian-condition}--\ref{asm: GD-CV-error} are likely satisfied. Suppose the data is generated from the generalized linear model with a canonic link function \citep{mccullagh1989generalized}. Specifically, suppose $Y_i$ has density $\exp(\eta_i Y_i - \phi(\eta_i)) h(Y_i)$, where $\phi: \bbR \mapsto \bbR$ is a given function and $\eta_i$ is linked with $X_i$ via $\eta_i = X_i^\top \theta^*$. Then given i.i.d. data $\calZ = \{ (X_i, Y_i) \}_{i=1}^n$, the gradient, and the Hessian of the negative log-likelihood function with ridge regularizer are given as 
	\begin{equation*}
		\begin{split}
			\nabla_\theta F(\calZ; \theta) &= -  \sum_{i=1}^n \left( Y_i - \phi'(\theta^\top X_i)  \right) X_i + \lambda \theta, \\
			\nabla^2_{\theta} F(\calZ; \theta) & =  \sum_{i=1}^n \phi''(\theta^\top X_i) X_i X_i^\top + \lambda \I_p.
		\end{split}
	\end{equation*} 
	Since $\bbE (Y_i) = \phi'(\theta^{*\top} X_i)$ and $\var(Y_i) = \phi''(\theta^{*\top} X_i)> 0$, we expect the first and second derivatives of the objective will be well-controlled when $\theta$ is close to $\theta^*$, the $X_i$s are bounded, and $n \geq p$. For many non-convex problems of interest, if we could have a warm-start initialization, then the problem at local typically satisfies certain restricted strong convexity and smoothness properties, and the Hessians and gradients along the iterates are often well-controlled as well \citep{loh2013regularized,chi2019nonconvex}.

\subsection{Guarantees for SGD}
In the SGD setting, we assume the batches $\{S_t\}_{t \geq 1}$ are drawn i.i.d. such that each data point is included in $S_t$ with probability $K/n$. To establish the approximation guarantee for IACV in this setting, we need the following three assumptions, which are stochastic analogues of Assumptions \ref{asm: hessian-condition}--\ref{asm: third-derivative-control}.
\begin{assumption}\label{asm: sgd-hessian-condition} 
	Suppose $F(\cdot; \theta)$ is twice differentiable in $\theta$ and there exists $\lambda_0, \lambda_1 > 0$ such that for all $t \geq 1$, $i \in [n]$, and $\alpha_t \leq \frac{1}{K \lambda_1}$, we have
 \begin{equation*}
 \begin{split}
        &\bbE (\|  [\I_p - \alpha_t \nabla^2_\theta F (\calZ_{S_t\setminus \{i \} }; \htheta^{(t-1)} )  ][\htheta^{(t-1)} - \htheta^{(t-1)}_{-i}]  \|_2 )
        \leq  (1 - \alpha_t K \lambda_0 ) \bbE \left[ \|\htheta^{(t-1)} - \htheta^{(t-1)}_{-i} \|_2 \right],
 \end{split}
 \end{equation*} where the expectation is taken over randomly drawn batches $S_1, \ldots, S_t$.  
\end{assumption} 
\noindent For example, this assumption would be satisfied under the stronger assumption that for all $ v, \theta \in \bbR^p$, all $i\in[n]$, and for a randomly drawn batch $S$,
 \begin{equation*}
     \bbE_S (\|  (\I_p - \alpha_t \nabla^2_\theta F (\calZ_{S\setminus \{i \} }; \theta )  )v  \|_2 ) \leq (1 - \alpha_t K \lambda_0) \|v\|_2.
 \end{equation*}
 Essentially, we can interpret $K \lambda_0$ and $K \lambda_1$ as bounding the smallest and largest eigenvalues of $ \nabla^2_\theta F (\calZ_{S \setminus \{i\} }; \theta ) $. 

Our next assumptions bound the gradient and the Lipschitz constant of the Hessian.
\begin{assumption} \label{asm: sgd-gradient-control} 
	For all $ t \geq 1$ and $i \in [n]$, we have $ \bbE_{S_1,\ldots,S_{t-1}} \| \nabla_\theta \ell(Z_i; \htheta^{(t-1)} ) \|_2 \leq \eta_i $ for some $\eta_i > 0$. \end{assumption}

\begin{assumption}\label{asm: sgd-third-derivative-control} 
There exists $\gamma > 0$ such that for any $i \in [n]$ and $t \geq 1$,
\begin{equation*}
\begin{split}
	&\bbE \sup_{\theta'}\Big\|  (\nabla^2_\theta F (\calZ_{S_t\setminus \{ i\}}; \theta' ) - \nabla^2_\theta F (\calZ_{S_t\setminus \{ i\}}; \htheta^{(t-1)} ) )[\htheta^{(t-1)}_{-i} - \htheta^{(t-1)}] \Big\|_2  \leq \gamma K \bbE \|\htheta^{(t-1)}_{-i} - \htheta^{(t-1)}\|_2^2,
\end{split}
\end{equation*} where the supremum is taken over $\theta' = a \htheta^{(t-1)}_{-i} + (1-a) \htheta^{(t-1)}$ for any $a \in [0,1]$.
\end{assumption}
\noindent For example,  this assumption would be satisfied if given any $v, \theta_1, \theta_2 \in \bbR^p$,
\begin{equation*}
    \begin{split}
        &\bbE_S \| (\nabla^2_\theta F (\calZ_{S\setminus \{ i\}}; \theta_1 ) - \nabla^2_\theta F (\calZ_{S\setminus \{ i\}}; \theta_2 )v  \|_2 
        \leq  \gamma K \|v\|_2 \|\theta_1 - \theta_2\|_2.
    \end{split}
\end{equation*}

 In addition, we need an extra assumption as follows.
\begin{assumption} \label{asm: reverse-jensen}
	There exists $\beta \geq 1$ such that for all $t \geq 1$, we have $\bbE_{S_1, \ldots, S_{t-1}} \| \htheta^{(t-1)}_{-i} - \htheta^{(t-1)} \|_2^2 \leq \beta \frac{n}{K} \left( \bbE_{S_1,\ldots, S_{t-1}} \| \htheta^{(t-1)}_{-i} - \htheta^{(t-1)} \|_2  \right)^2$.
\end{assumption}
	 Assumption \ref{asm: reverse-jensen} can be viewed as a reverse Jensen's inequality, but there is an inflation factor $\beta n/K$ on the right-hand side. The order $n/K$ of this inflation factor is expected. A simple way to reveal this is to examine the inequality when $t = 2$. If $\hthetai^{(0)} = \htheta^{(0)}$,  then $\htheta^{(1)}_{-i} - \htheta^{(1)} = 0$ with probability $(1- K/n)$ and $\htheta^{(1)}_{-i} - \htheta^{(1)} = \alpha_1 \nabla_\theta \ell(Z_i; \htheta^{(0)} )$ with probability $K/n$ (depending on whether data point $i$ is excluded or included in the first batch $S_1$ at time $t=1$). Thus, we have $\bbE_{S_1} \| \htheta^{(1)}_{-i} - \htheta^{(1)} \|_2^2 = \frac{n}{K}(\bbE_{S_1} \| \htheta^{(1)}_{-i} - \htheta^{(1)} \|)^2 $. 

Now we are ready to present a guarantee on the approximation error of IACV in the SGD setting.
\begin{theorem}[Approximation Error of IACV for SGD] \label{th: SGD-case-approx-error}
	 Suppose $\htheta^{(0)} = \htheta^{(0)}_{-i} = \ttheta^{(0)}_{-i}$ for all $i \in [n]$, $\alpha_t \leq 1/ (K \lambda_0)$ for all $t \geq 1$ and Assumptions \ref{asm: sgd-hessian-condition}--\ref{asm: reverse-jensen} are satisfied with $\gamma \beta < 2 \lambda_0$ and $K \geq \frac{4 \|\eta\|_\infty }{2 \lambda_0 - \gamma \beta}$. Then for all $t \geq 1$, $i \in [n]$, we have \[\bbE \| \htheta^{(t)} - \htheta^{(t)}_{-i} \|_2 \leq \frac{2 \eta_i}{(2 \lambda_0 - \gamma\beta) n}\] and \[\bbE \|\ttheta_{-i}^{(t)} - \htheta^{(t)}_{-i}\|_2 \leq \frac{4 \gamma \beta \eta_i^2 }{\lambda_0 (2 \lambda_0 - \gamma\beta)^2 nK },\] where the expectation is taken over the randomly drawn batches $S_1,\dots,S_t$.
\end{theorem}
\noindent By Theorem \ref{th: SGD-case-approx-error}, the expected approximation error in the SGD case is bounded as \[\bbE \left(\textnormal{Err}_{\textnormal{approx}}^{(t)} \right)\leq 
\frac{4 \gamma \beta \|\eta\|^2_\infty}{\lambda_0 (2 \lambda_0 - \gamma)^2 n^2 }\textnormal{ for all $t\geq 1$.}\]

Interestingly, we find that the ``baseline'' approximation error $\bbE \| \htheta^{(t)} - \htheta^{(t)}_{-i} \|_2$ does not depend on $K$ and is at the same order $1/n$ as the one in Theorem \ref{th: GD-case-approx-error}. In contrast, the IACV error $\bbE \|\ttheta_{-i}^{(t)} - \htheta^{(t)}_{-i}\|_2$ scales at the order of $1/nK$, as compared to order $1/n^2$ in the previous setting. This factor comes from the inflation factor $n/K$ in Assumption \ref{asm: reverse-jensen}; as discussed immediately below this assumption, the factor $n/K$ cannot be removed if we expect the assumption to hold in practice for all $t\geq 1$, but a remaining open question is whether the factor $n/K$ can be removed if we only require it to hold for all $t\geq T_0$ for some large $T_0$. If this is the case, then (for sufficiently large $t$) the IACV error in Theorem~\ref{th: SGD-case-approx-error} will scale as $1/n^2$ rather than $1/nK$.

Finally, we provide a CV error guarantee under the SGD setting, with one additional assumption: 
\begin{assumption} \label{asm:SGD-CV-error}
    Suppose for any $i \in [n]$ and all $t \geq 1$, 
    there exists $\eta'_i > 0$ such that
  \begin{equation*}
     \bbE \left\{ \sup_{a\in[0,1]} |\nabla_\theta \ell(Z_i; a \ttheta^{(t)}_{-i} + (1-a) \htheta^{(t)}_{-i}  )^\top ( \ttheta^{(t)}_{-i}  -   \htheta^{(t)}_{-i})| \right\} \leq \eta'_i \bbE \|\ttheta^{(t)}_{-i}  -   \htheta^{(t)}_{-i}\|_2.
  \end{equation*}
\end{assumption}
\noindent For example, this assumption will be satisfied if $\ell(Z_i,\cdot)$ is $\eta'_i$-Lipschitz for each $i$.
\begin{theorem}[CV Error of IACV for SGD] \label{th: model-assessment-bound-SGD} 
If the assumptions in Theorem \ref{th: SGD-case-approx-error} and Assumption \ref{asm:SGD-CV-error} are satisfied, then for all $t \geq 1$:
   \begin{equation*}
 \begin{split}
  \bbE\left(\textnormal{Err}_{\textnormal{CV}}^{(t)} \right) \leq \frac{1}{n} \sum_{i=1}^n
 \frac{4\gamma \beta \eta'_{i} \eta_{i}^2  }{\lambda_0 (2 \lambda_0 - \gamma\beta)^2 nK },
 \end{split}
 \end{equation*} where the expectation is taken over the randomly drawn batches $S_1,\dots,S_t$.
\end{theorem}
\noindent As before, this result is 
analogous to Theorem~\ref{th: model-assessment-bound} for the GD setting, but with rate $1/nK$ in place of $1/n^2$.

\section{Limiting Behavior of IACV} \label{sec: limiting-behavior}
As we have seen in Section \ref{sec: approximation-property}, IACV achieves accurate approximation along all iterations of the algorithm under some regularity conditions. Thus, in the setting when $\htheta^{(t)}$ indeed converges to $\htheta$, IACV also has the same approximation properties in the limit, and its approximation error bound is comparable to the guarantee of the one-step Newton (NS) estimator in \eqref{eqn:NS-intro} as we have mentioned in Section \ref{sec: guarantee-for-GD}. In this section, we show this is not a coincidence---in fact, there is a close connection of the proposed estimator $\ttheta^{(t)}_{-i}$ to the NS estimator when $t \to \infty$ and the algorithm converges to $\htheta$.

\begin{theorem}
\label{th: limiting-behavior-estimator-GD} 
Suppose $\htheta^{(t)}$ converges to $\htheta$, Assumption \ref{asm: hessian-condition} is satisfied along all iterations, $\alpha_t = \alpha < 1/(n \lambda_1) $ for all $t \geq 1$, and $\nabla^2_\theta F(\calZ_{-i}; \theta) $, $ \nabla_\theta F (\calZ_{-i}; \theta)$ are continuous in $\theta$ for all $i \in [n]$. Then $\ttheta^{(t)}_{-i}$ defined in \eqref{eqn: iter-IACV-GD} converges to $\ttheta^{\NS}_{-i}$ in \eqref{eqn:NS-intro}, as $t\rightarrow\infty$.
\end{theorem}

Moreover, in Theorem \ref{th: limiting-behavior-estimator-prox-GD}, we will show similar guarantees continue to hold in the ProxGD setting. (In the SGD setting, however, the techniques for proving Theorem \ref{th: limiting-behavior-estimator-GD} fail---we will give an intuition for why this is the case, in Appendix \ref{proof-sec: limiting-behavior}.)

In view of the results in Theorem \ref{th: limiting-behavior-estimator-GD}, we can regard IACV as an extension of $\ttheta^{\NS}_{-i}$ in \eqref{eqn:NS-intro} to iterative solvers with provable per-iteration guarantees. This provides a safe and efficient way to approximate CV in practice where it is often agnostic whether the algorithm will converge to the solution of ERM or not.

\section{Simulation Studies} \label{sec: simulation}
In this section, we conduct numerical studies to investigate the empirical performance of the proposed IACV method and to verify our theoretical findings for GD, SGD, and ProxGD.\footnote{Code to reproduce all experiments is available at \url{https://github.com/yuetianluo/IACV}.} The data is generated from a logistic regression model, with $Z_i=(X_i,Y_i)\in\bbR^p\times\{0,1\}$, with dimension $p=20$ and with $X_i$ drawn with i.i.d.~$N(0,1)$ entries, while 
\[Y_i \sim\textnormal{Bernoulli}(\exp(X_i^\top\theta^*)/(1+\exp(X_i^\top\theta^*)))\]
for true parameter vector $\theta^*$ which has 5 randomly chosen nonzero entries drawn as $N(0,1)$, and all other entries zero. 
Our objective function is given by regularized negative log-likelihood, $F(\calZ;\theta) = \sum_{i=1}^n \ell(Z_i; \theta) + \lambda\pi(\theta)$, where
\begin{equation*}
	\ell(Z_i; \theta) =  -  Y_i \cdot X_i^\top \theta + \log( 1+\exp(X_i^\top\theta)).
\end{equation*}
For GD and SGD, we use ridge regularization, with $\pi(\theta)=\|\theta\|^2_2$ and penalty parameter $\lambda=10^{-6}\cdot n$. For ProxGD we instead use the logistic Lasso, with $\pi(\theta)=\|\theta\|_1$ and $\lambda =10^{-6}\cdot n$. We initialize the algorithm at the origin.
Each simulation study is repeated for $100$ independent trials.

Our proposed method is given by the IACV estimator $\ttheta^{(t)}_{-i}$, defined in~\eqref{eqn: iter-IACV-GD} for GD and in~\eqref{eqn: iter-IACV-SGD} for SGD. For comparison,
we also implement the one-step Newton (NS) and infinitesimal jackknife (IJ) estimators along the optimization path, i.e., we use the $t$th iteration $\htheta^{(t)}$ in place of the true minimizer $\htheta$ in the definition of the NS~\eqref{eqn:NS-intro} or IJ~\eqref{eqn:IJ-intro} estimators, leading to the approximate NS estimator
$$\ttheta^{\NS(t)}_{-i} = \htheta^{(t)} - \left(\nabla^2_{\theta} F (  \calZ_{-i}; \htheta^{(t)} ) \right)^{-1} \nabla_{\theta} F ( \calZ_{-i}; \htheta^{(t)} )$$
and similarly the approximate IJ estimator
$$\ttheta^{\IJ(t)}_{-i} = \htheta^{(t)} - \left(\nabla^2_{\theta} F (  \calZ; \htheta^{(t)} ) \right)^{-1} \nabla_{\theta} F ( \calZ_{-i}; \htheta^{(t)} ).$$
Finally, we also compare to a ``baseline'' estimator where we simply approximate the leave-one-out iterate $\htheta^{(t)}_{-i}$ with the full-data iterate $\htheta^{(t)}$:
\begin{equation}\label{eq: baseline}
\htheta^{\textnormal{baseline}(t)}_{-i} = \htheta^{(t)}.\end{equation}
Of course, this baseline is not useful in practice (since data point $i$ has not actually been removed from the estimator), and thus any proposed method is only meaningful if it can perform substantially better than this baseline.
We measure the accuracy of the four methods (baseline, NS, IJ, and IACV) in terms of the averaged approximation error $\textnormal{Err}_{\textnormal{approx}}^{(t)}$ defined in~\eqref{eqn: define-approx-error}, and the relative CV error, $\textnormal{RelErr}^{(t)}_{\textnormal{CV}} = \textnormal{Err}_{\textnormal{CV}}^{(t)}/
\CV(\{\htheta^{(t)}_{-i}\}_{i=1}^n)$, where $\textnormal{Err}_{\textnormal{CV}}^{(t)}$ is defined in~\eqref{eqn: define-model-assessment-error}.

\subsection{Gradient Descent (GD)} \label{sec: simulation-GD}
For the gradient descent simulation, we consider sample sizes $n=250$ and $n=1000$, and $\alpha_t = 0.5/n$.
In the top panels of Figure \ref{fig: approximation-error-GD}, we can see that along the iterations of the algorithm, the approximation error of IACV is always better than NS and IJ before the convergence of the algorithm. As a result of that, IACV also achieves better CV error as we illustrate in the middle panels of Figure \ref{fig: approximation-error-GD}. In fact, during early iterations, the error of both NS and IJ is higher than the (noninformative) baseline method~\eqref{eq: baseline}, while IACV's error is substantially lower, and also shows a smaller variance in estimation as it has narrower shaded areas. When the algorithm converges,  we find that the performance of IACV and the NS estimator are almost the same, which matches the theoretical prediction in Theorem \ref{th: limiting-behavior-estimator-GD} (where we see that, under mild assumptions, IACV will converge to the NS estimator, i.e., the black line and the red line will meet in the limit), while the IJ method shows higher error even at convergence. 

In addition, we observe that as sample size $n$ increases from $250$ to $1000$, the limiting approximation error of IACV decreases from $1.5 \times 10^{-3}$ to $6.8 \times 10^{-5}$. This roughly matches what we have shown in Theorem \ref{th: GD-case-approx-error} that the approximation error decreases quadratically with respect to the sample size. 

Finally, in the bottom panels of Figure \ref{fig: approximation-error-GD}, we report the runtime of IACV as compared to the exact leave-one-out CV. We can see that IACV is much faster than the exact leave-one-out CV method; in particular when $n = 1000$, IACV shows approximately 6-7 times speed-up. A larger scale simulation for GD is provided in Appendix \ref{app: large-scale-gd}.
\begin{figure}[t]
	\centering
	\includegraphics[width = 0.85\textwidth]{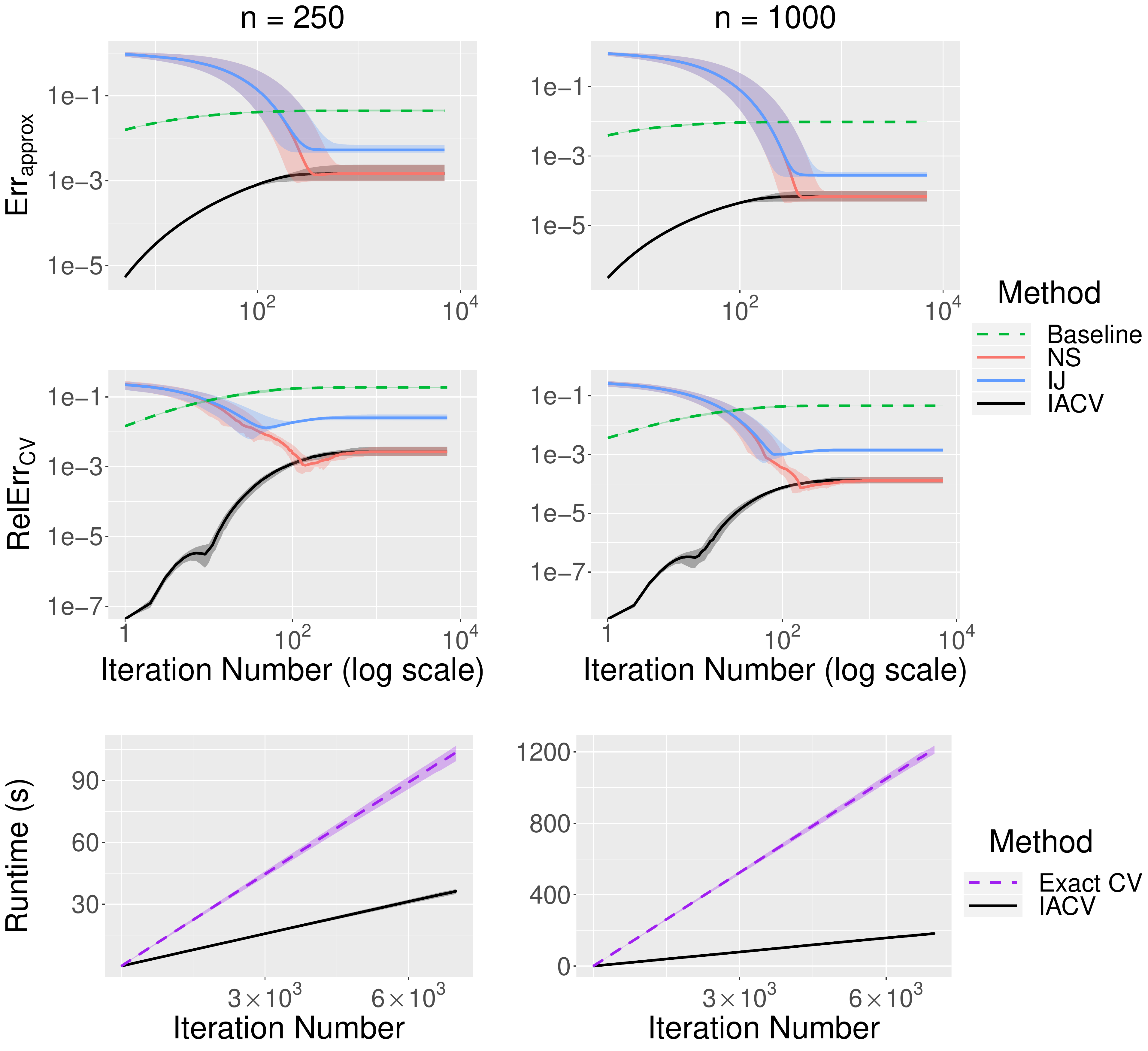} 
	\caption{Comparison of the baseline, NS, IJ, and IACV methods in gradient descent for logistic regression with ridge regularizer. Top: approximation error comparison; middle: relative CV error comparison; bottom: runtime comparison of exact leave-one-out CV and IACV.  Solid lines represent the median value over $100$ experiments and shaded areas denote the region between lower and upper quartiles.} \label{fig: approximation-error-GD}
\end{figure}

\subsection{Stochastic Gradient Descent (SGD)} \label{sec: simulation-SGD}

For the stochastic gradient descent simulation, we take sample size $n=1000$, and test batch size $K=100$ and $K=400$. We choose $\alpha_t$ based on a common strategy called ``epoch doubling'' in the literature, where we run $T_0=1000$ steps with step size $\alpha =0.5/K$, then run $2T_0$ steps with step size $\alpha/2$, and so on.

We plot the accuracy of the different methods in Figure \ref{fig: approximation-error-SGD}, in the top (approximation error) and middle (CV error) panels. We can see that IACV has a clear advantage over the other methods. In contrast to the simulation results in Figure \ref{fig: approximation-error-GD} for GD, here the red line (for NS) does not reach the black line (for IACV) even after $T=10^5$ iterations, in terms of approximation error. This may be due to the slow convergence of SGD;
nonetheless, it is still unclear whether the NS estimator $\ttheta^{\NS(t)}_{-i}$ and our estimator $\ttheta^{(t)}_{-i}$ will converge to the same limit or not, since we do not know whether a result analogous to Theorem \ref{th: limiting-behavior-estimator-GD} holds in the setting of SGD. As for GD, we see that IACV offers error far lower than both NS and IJ during early iterations, when NS and IJ show error higher even than the baseline estimator.  

Finally, we show the runtime of IACV with exact leave-one-out CV in the bottom panel in Figure \ref{fig: approximation-error-SGD}. We can see IACV still has clear computational advantages in this setting. 

\begin{figure}[t]
	\centering
	\includegraphics[width = 0.85\textwidth]{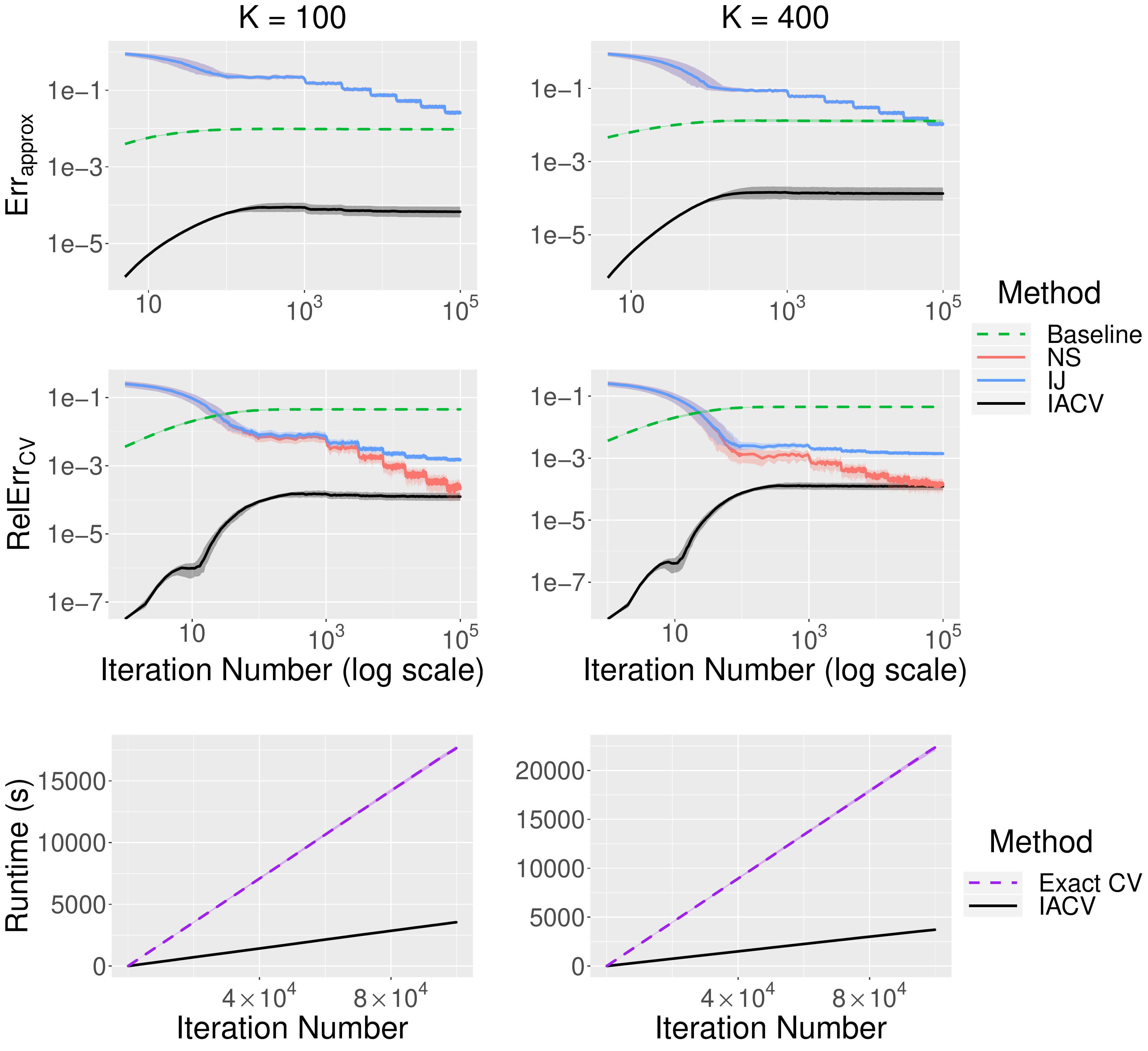} 
	\caption{Comparison of the baseline, NS, IJ, and IACV methods in stochastic gradient descent for logistic regression with ridge regularizer; details of the panels are the same as in Figure~\ref{fig: approximation-error-GD}. Note that in the top panels, the results for NS and IJ (red and blue) are nearly perfectly overlapping and thus difficult to distinguish.
 } \label{fig: approximation-error-SGD}
\end{figure}

\subsection{Proximal Gradient Descent (ProxGD)} \label{sec: simulation-ProxGD}

Finally, for the proximal gradient descent simulation, we take sample sizes $n=250$ and $n=1000$, and $\alpha_t = 0.5/n$, as for the GD simulation. The regularizer is now $\pi(\theta)=\|\theta\|_1$, a nonsmooth function (see Appendix~\ref{app: prox-GD-theory} for the definition of the NS and IJ methods in this nonsmooth setting). 

The results for this setting are qualitatively very similar to the GD setting; again, we see that IACV shows good accuracy in terms of both approximation error and CV error even at early iterations, while NS and IJ show error higher than the baseline during early iterations. Finally, when the algorithm converges, the performance of our estimator is similar to the NS estimator as predicted in Theorem \ref{th: limiting-behavior-estimator-prox-GD} (which is the analogue of Theorem~\ref{th: limiting-behavior-estimator-GD}, for the ProxGD setting).

\begin{figure}[t]
	\centering
	\includegraphics[width = 0.85\textwidth]{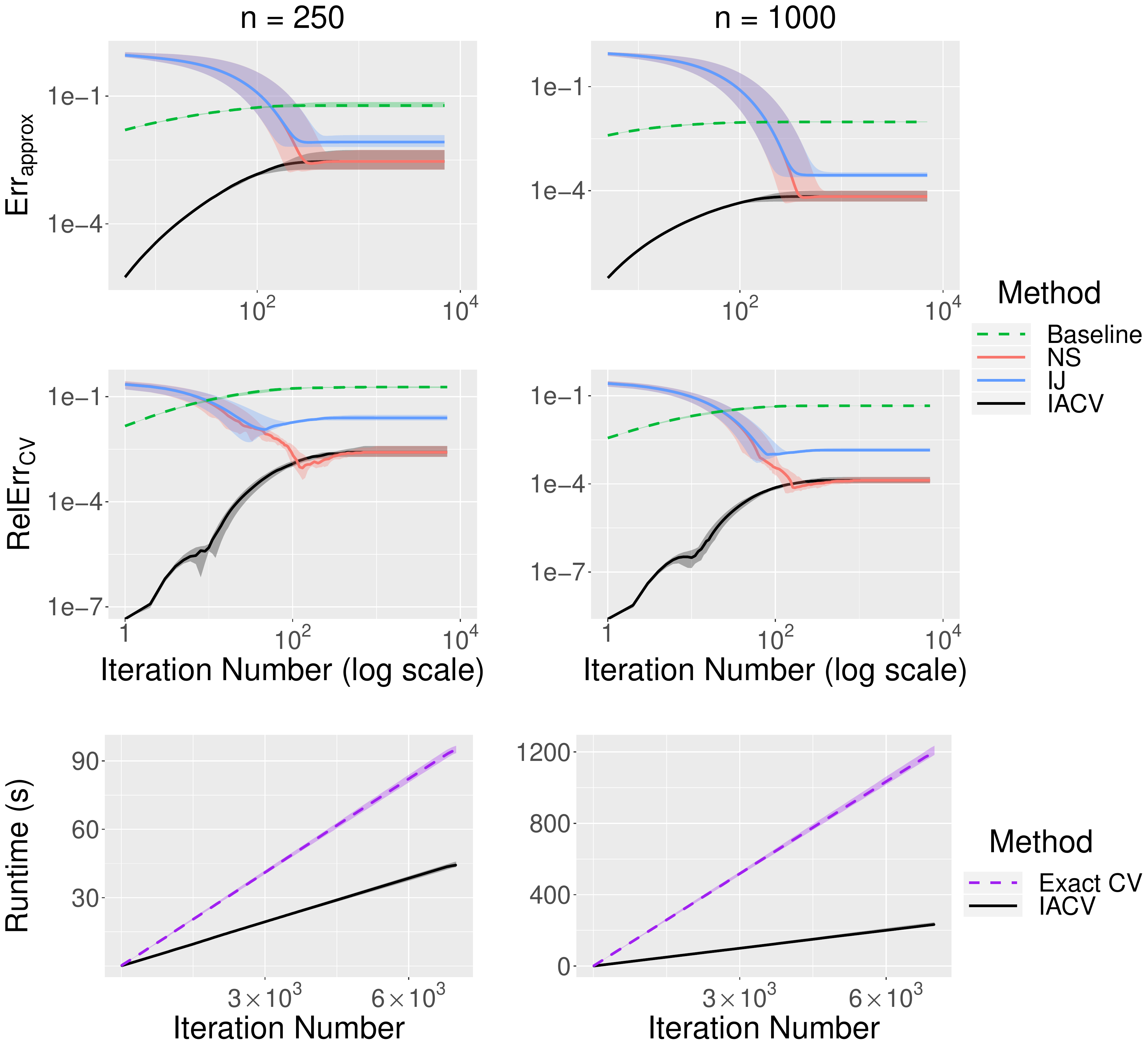} 
	\caption{Comparison of the baseline, NS, IJ, and IACV methods in proximal gradient descent for logistic regression with $\ell_1$ regularizer (the logistic Lasso); details of the panels are the same as in Figure~\ref{fig: approximation-error-GD}.} \label{fig: approximation-error-proxGD}
\end{figure}

 \section{Conclusion and Discussions}\label{sec: conclusion}
In this paper, we provide a new method, iterative approximate cross-validation (IACV), to efficiently approximate the computationally expensive leave-one-out CV under the ERM framework when the problem is solved by common iterative algorithms. IACV achieves efficient CV approximation along the whole trajectory of the algorithm, which is in big contrast to the existing CV approximation methods.

This work suggests several interesting directions for further exploration. For example, Theorem \ref{th: limiting-behavior-estimator-GD} establishes that, for gradient descent, the IACV and NS estimators coincide at convergence; we do not yet know whether an analogous result holds for SGD. We can also consider the performance of IACV in a high-dimensional setting. For instance, in regression with $p>n$,
the NS estimator defined in \eqref{eqn:NS-intro} fails due to high dimensionality, but \citet{stephenson2020approximate} propose running $\ell_1$-regularized regression and then running the NS estimator on the selected active set of $\htheta$. It would be interesting to see whether we can adapt the IACV method to that setting and provide theoretical guarantees for the sparse high-dimensional regime.

\section*{Acknowledgements}
Z.R. and R.F.B were supported by the Office of Naval Research via grant N00014-20-1-2337. R.F.B. was additionally supported by the National Science Foundation via grants DMS-1654076 and DMS-2023109.

\bibliographystyle{apalike}
\bibliography{ref}

\appendix
\section{Guarantees for IACV in the Proximal GD Setting} \label{app: prox-GD-theory}
We now provide the analogues of Theorems \ref{th: GD-case-approx-error} and \ref{th: model-assessment-bound} in the ProxGD setting, working with the objective function $F(\calZ;\theta) =g(\calZ;\theta)+h(\theta)$, where as before, $g(\calZ;\theta)=\sum_{i=1}^n\ell(Z_i;\theta)$ is the empirical risk and $h(\theta)=\lambda\pi(\theta)$ is the nonsmooth regularization term.

\begin{theorem}
\label{th: Prox-GD-case-approx-error}
	  Suppose $\htheta^{(0)} = \htheta^{(0)}_{-i} = \ttheta^{(0)}_{-i}$ for all $i \in [n]$, $\alpha_t \leq 1/(n \lambda_1)$ for $t \geq 1$, and Assumptions \ref{asm: hessian-condition}--\ref{asm: third-derivative-control} are satisfied with $F(\cdot; \theta)$ replaced by $g(\cdot; \theta)$ and with $\gamma < 2 \lambda_0$ and $n \geq \frac{4 \|\eta\|_\infty }{2 \lambda_0 - \gamma} $. Assume also that $h$ is convex. Then for all $ t \geq 1$ and $i \in [n]$, we have
  \[\| \htheta^{(t)} - \htheta^{(t)}_{-i} \|_2 \leq \frac{2 \eta_i}{(2 \lambda_0 - \gamma) n}\] and\[\|\ttheta_{-i}^{(t)} - \htheta^{(t)}_{-i}\|_2 \leq \frac{4 \gamma \eta_i^2 }{\lambda_0 (2 \lambda_0 - \gamma)^2 n^2 }.\]
\end{theorem}
\noindent Note that these bounds are identical to the results obtained in Theorem~\ref{th: GD-case-approx-error}, and thus as before, we obtain the approximation error bound
\[\textnormal{Err}_{\textnormal{approx}}^{(t)}\leq 
\frac{4 \gamma \|\eta\|^2_\infty}{\lambda_0 (2 \lambda_0 - \gamma)^2 n^2 }\textnormal{ for all $t\geq 1$.}\]
The proof of Theorem~\ref{th: Prox-GD-case-approx-error} is provided in Appendix \ref{proof-sec: approximation-property}.

Next we bound the CV error for IACV in the ProxGD setting.
\begin{theorem}
\label{th: model-assessment-bound-proxgd}
  Suppose the assumptions in Theorem \ref{th: Prox-GD-case-approx-error} and Assumption \ref{asm: GD-CV-error} are satisfied. Then for all $t \geq 1$,
  \[\textnormal{Err}_{\textnormal{CV}}^{(t)} \leq \frac{1}{n} \sum_{i=1}^n \frac{4 \gamma \eta'_i \eta_i^2}{\lambda_0 (2 \lambda_0 - \gamma)^2 n^2 }.\]
\end{theorem}
\noindent Again, this bound is the same as the one obtained in Theorem~\ref{th: model-assessment-bound} for GD.
The proof of Theorem \ref{th: model-assessment-bound-proxgd} is essentially the same as the proof of Theorem \ref{th: model-assessment-bound} and for simplicity, we omit it here.

Next, we provide convergence theory for IACV in the ProxGD setting, to obtain a result analogous to Theorem~\ref{th: limiting-behavior-estimator-GD} for the GD case comparing the IACV and NS estimators.
First, we need to define the NS estimator in this setting where $F(\calZ;\theta)$ is nonsmooth. The main idea is to replace the one Newton step in \eqref{eqn:NS-intro} with one proximal Newton step \citep{wilson2020approximate}:
\begin{equation}  \label{eq: one-step-prox-Newton-estimator}
\begin{split}
	\ttheta^{\NS}_{-i} &= \prox_{h}^{\nabla^2_\theta g( \calZ_{-i}; \htheta)} \left( \htheta - \left(\nabla^2_{\theta} g( \calZ_{-i}; \htheta) \right)^{-1} \nabla_{\theta} g (\calZ_{-i}; \htheta ) \right),
\end{split}
\end{equation} where given any positive definite matrix $H \in \bbR^{p \times  p}$ and convex function $h: \bbR^p \to \bbR$, $\prox_{h}^{H} (x)$ is defined as follows: 
\begin{equation} \label{eq: prox-bar-definition}
	\prox_{h}^{H} (x) = \arg \min_{z \in \bbR^p} \left\{\frac{1}{2}  (x - z)^\top H (x - z) + h(z)\right\}. 
\end{equation}
(We can similarly define the IJ estimator in this setting as
\begin{equation}\label{eq: IJ-prox}\ttheta^{\IJ}_{-i} = \prox_{h}^{\nabla^2_\theta g( \calZ_{-i}; \htheta)} \left( \htheta - \left(\nabla^2_{\theta} g( \calZ; \htheta) \right)^{-1} \nabla_{\theta} g (\calZ_{-i}; \htheta ) \right).\end{equation}
 In our simulations in Section~\ref{sec: simulation-ProxGD}, we simply replace $\htheta$ with $\htheta^{(t)}$ in equations~\eqref{eq: one-step-prox-Newton-estimator} and~\eqref{eq: IJ-prox} to obtain our approximate NS and IJ iterations.)

Next, we show $\ttheta^{(t)}_{-i}$ in \eqref{eqn: iter-IACV-ProxGD} will converge to $\ttheta^{\NS}_{-i}$ in \eqref{eq: one-step-prox-Newton-estimator}, under proper assumptions.
\begin{theorem}
\label{th: limiting-behavior-estimator-prox-GD} 
Suppose $\htheta^{(t)}$ converges to $\htheta$, Assumption \ref{asm: hessian-condition} is satisfied along all iterations with $g(\cdot; \theta)$ in place of $F(\cdot; \theta)$, $\alpha_t = \alpha < 1/(n \lambda_1) $ for all $t \geq 1$, and $\nabla^2_\theta g(\calZ_{-i}; \theta)$, $\nabla_\theta g (\calZ_{-i}; \theta)$ are continuous in $\theta$ for all $i \in [n]$. In addition, we assume $h(\theta)$ is convex. Then it holds that $\ttheta^{(t)}_{-i}$ in \eqref{eqn: iter-IACV-ProxGD} converges to $\ttheta^{\NS}_{-i}$ in \eqref{eq: one-step-prox-Newton-estimator}, as $t\rightarrow\infty$.
\end{theorem}
The proof of this theorem is provided in Appendix \ref{proof-sec: limiting-behavior}.

\section{Details for deriving computational complexity}\label{app:computational-complexity}
To derive the calculations in Table~\ref{tab: time-memory-complexity}, for example, for GD,
the main per-iteration cost of our method comes from computing $\{\nabla_\theta g(\calZ_{-i}; \htheta^{(t-1)}) \}_{i=1}^n$ and $\{\nabla^2_\theta g(\calZ_{-i}; \htheta^{(t-1)})\}_{i=1}^n$, and performing the Hessian-gradient product. As we assume above, due to the relation $\nabla_\theta g(\calZ_{-i}; \htheta^{(t-1)}) =   \sum_{j \in [n],j \neq i} \nabla_\theta \ell(Z_j; \htheta^{(t-1)})$, the cost of computing $\{\nabla_\theta g(\calZ_{-i}; \htheta^{(t-1)}) \}_{i=1}^n$ and $\{\nabla^2_\theta g(\calZ_{-i}; \htheta^{(t-1)})\}_{i=1}^n$ is of order $n(A_p+B_p)$. Moreover, the cost of the Hessian-gradient product is of order $np^2$. On the other hand, the main per-iteration computational cost of the exact leave-one-out CV (i.e., running an iteration of gradient descent on the dataset $\calZ_{-i}$, for each $i$), comes from evaluating $\nabla_\theta g(\calZ_{-i}, \htheta^{(t-1)}_{-i} )$ and then subtracting this gradient from the current estimate, for each $i$. These two steps have cost of order $nA_p$ and $p$, respectively, for each $i$, and therefore the total cost of one iteration has order $n^2A_p+np$. Similar calculations can be performed to compute the order of computational complexity for SGD and ProxGD as well.

Note here that we do not intend to compare the runtime of NS and IJ with IACV and exact CV because these methods are not comparable in terms of their target problem. NS and IJ are one-step methods, where given a single solution (i.e., a single $\widehat\theta^{(T)}$ approximating $\widehat\theta$), we run the method once to approximate the leave-one-out models. On the other hand, exact iterative CV as well as our IACV methods are both performed in an online fashion, with steps carried out for each $t=1,\dots,T$. The cost of NS and IJ will not scale with $T$, while naturally cost of CV and IACV must scale with $T$; this apparent computational benefit of NS and IJ is simply due to the fact that NS and IJ ignore the iterative nature of the algorithm and thus are completely invalid (i.e., errors are higher than the noninformative ``baseline'') for early, pre-convergence times $T$.

\section{A Larger Scale Simulation for IACV in GD} \label{app: large-scale-gd}
 In Figures \ref{fig: gd_large} and \ref{fig: gd_large_runtime}, we provide simulation results for comparing IACV and other methods when $n = 5000$, $p = 50$. We observe a similar pattern as in the existing plots: the approximation error of IACV is always better than NS and IJ before the convergence of the algorithm and IACV is much faster than the exact leave-one-out CV. If we further grow $p$ and $n$, we find it is too expensive to run the whole program as the exact leave-one-out CV takes too much time to run for many independent trials, but we nonetheless would expect to see similar performance.

\begin{figure}[t]
	\centering
	\includegraphics[width = 0.82\textwidth]{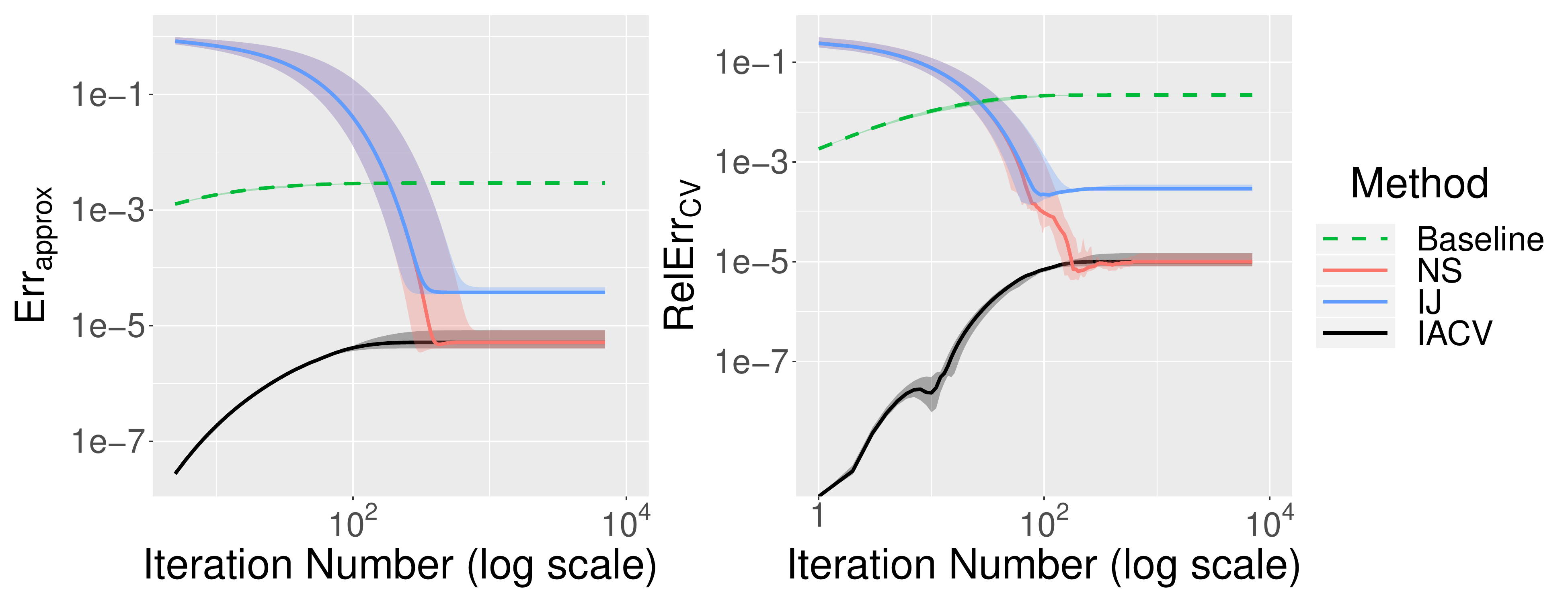} 
	\caption{Approximation error and relative CV error comparison of the baseline, NS, IJ, and IACV methods in gradient descent for logistic regression with ridge regularizer. Solid lines represent the median value over $100$ experiments and shaded areas denote the region between lower and upper quartiles.} \label{fig: gd_large}
\end{figure}

\begin{figure}[t]
	\centering
	\includegraphics[width = 0.5\textwidth]{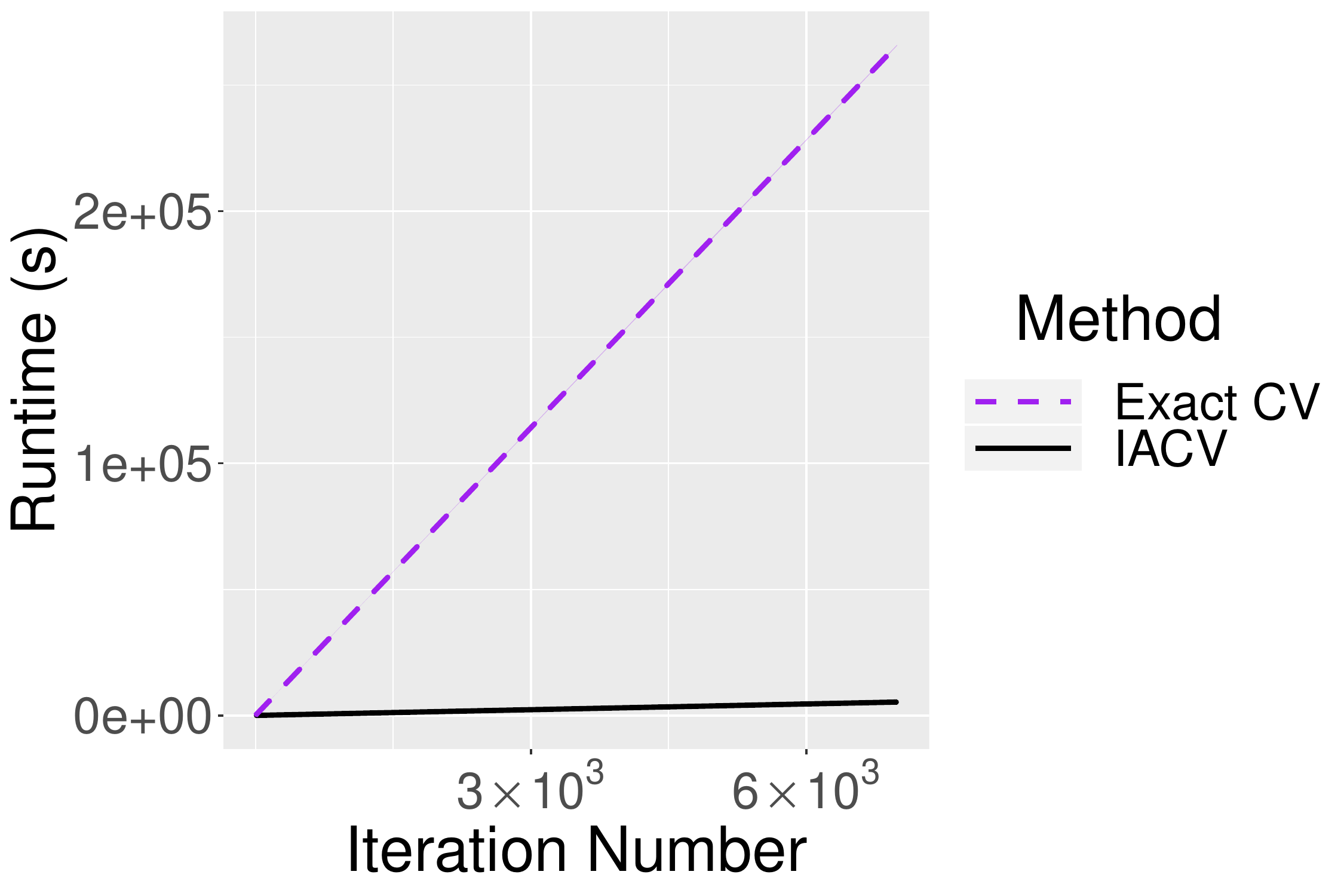} 
	\caption{Runtime comparison of the baseline, NS, IJ, and IACV methods in gradient descent for logistic regression with ridge regularizer.} \label{fig: gd_large_runtime}
\end{figure} 

\section{A Iteration-Dependent Upper Bound for $\|\ttheta_{-i}^{(t)}-\htheta^{(t)}_{-i}\|_2$} \label{app: a-time-dependent-bound}
Notice that in Theorem \ref{th: GD-case-approx-error}, we only present a version of the upper bound which is independent of $t$ for convenience. It is simple and clearly illustrates the idea that the error bound holds along the whole trajectory of the learning process. But based on our proof, a $t$-dependent iterative bound on the approximation error can be obtained and it will typically be sharper at the early stage of training. Specifically, in our proof of Theorem \ref{th: GD-case-approx-error}, \eqref{ineq: iteration-dependent-bound2} shows
\begin{equation*}
	\|\ttheta_{-i}^{(t)}-\htheta^{(t)}_{-i} \|_2 \leq (1 - \alpha_t n \lambda_0) \|\ttheta_{-i}^{(t-1)}-\htheta^{(t-1)}_{-i} \|_2 + \alpha_t n\gamma \|\htheta^{(t-1)}_{-i}-\htheta^{(t-1)} \|_2^2, \quad \forall t \geq 1.
\end{equation*} The error bound of $\ttheta_{-i}^{(t)}$ depends on error bound of $\ttheta_{-i}^{(t-1)}$ and $\|\htheta^{(t-1)}_{-i}-\htheta^{(t-1)} \|_2$. Furthermore, $\|\htheta^{(t-1)}_{-i}-\htheta^{(t-1)} \|_2$ also satisfies the following iterative upper bound by \eqref{ineq: iteration-dependent-bound1}:
\begin{equation*}
	\|\htheta^{(t-1)}_{-i}-\htheta^{(t-1)} \|_2 \leq (1- \alpha_{t-1} n \lambda_0) \|\htheta^{(t-2)} - \htheta^{(t-2)}_{-i}\|_2 + \alpha_{t-1} \eta_i  + \alpha_{t-1} n \gamma  \|\htheta^{(t-2)} - \htheta^{(t-2)}_{-i}\|_2^2.
\end{equation*}    
If we initialize $\htheta^{(0)} = \htheta^{(0)}_{-i} = \ttheta^{(0)}_{-i} = 0$, we can see the error bound for $\|\htheta^{(t-1)}_{-i}-\htheta^{(t-1)} \|_2$ will slightly increase from $0$ over the iteration $t$ and stabilize when the error contraction ($- \alpha_{t-1} n \lambda_0 \|\htheta^{(t-2)} - \htheta^{(t-2)}_{-i}\|_2$) cancels out the per-iteration approximation error accumulation ($\alpha_{t-1} \eta_i  + \alpha_{t-1} n\gamma \|\htheta^{(t-2)} - \htheta^{(t-2)}_{-i}\|_2^2$). By a similar argument, we also have the error bound for $\|\ttheta_{-i}^{(t)}-\htheta^{(t)}_{-i}\|_2$ will increase first and then stabilize.

\section{Proofs}

\subsection{Proofs of error bound results} \label{proof-sec: approximation-property}

In this section we prove the approximation error bounds and CV error bounds for GD, SGD, and ProxGD.

\begin{proof}[Proof of Theorem \ref{th: GD-case-approx-error}]

 {\bf Step 1.} In this step, we prove the bound on $\|\htheta^{(t)}-\htheta^{(t)}_{-i}\|_2$ by induction. By the update rule of  $\htheta^{(t)}$ and $\htheta^{(t)}_{-i}$, we have
	\begin{equation} \label{eq: quantity-1-main-bound}
	\begin{split}
		\htheta^{(t)} - \htheta^{(t)}_{-i} &= \htheta^{(t-1)} - \htheta^{(t-1)}_{-i} + \alpha_t \left( \nabla_\theta F (\calZ_{-i}; \htheta^{(t-1)}_{-i} ) - \nabla_{\theta} F(\calZ; \htheta^{(t-1)}  ) \right) \\
		&  = \htheta^{(t-1)} - \htheta^{(t-1)}_{-i} - \alpha_t \nabla \ell(Z_i;\htheta^{(t-1)} ) +   \alpha_t \left( \nabla_\theta F (\calZ_{-i}; \htheta^{(t-1)}_{-i} ) - \nabla_{\theta} F(\calZ_{-i}; \htheta^{(t-1)}  ) \right)  .
	\end{split}
	\end{equation} Given any $u \in \bbR^p$, by the Taylor expansion for $g_u(\theta) = \langle u,  \nabla_\theta F (\calZ_{-i}; \theta ) \rangle$, we have 
	\begin{equation} \label{eq: gradient-taylor-expansion}
	\begin{split}
	g_u(\htheta^{(t-1)}_{-i}) &= g_u( \htheta^{(t-1)} ) + \nabla_\theta g_u (\theta') \\
	\Longleftrightarrow \langle u,  \nabla_\theta F (\calZ_{-i}; \htheta^{(t-1)}_{-i} ) \rangle &= \langle u,  \nabla_\theta F (\calZ_{-i}; \htheta^{(t-1)} ) \rangle + \langle u, \nabla^2_\theta F (\calZ_{-i}; \htheta' )[\htheta^{(t-1)}_{-i} - \htheta^{(t-1)}] \rangle \\
	& = \langle u,  \nabla_\theta F (\calZ_{-i}; \htheta^{(t-1)} ) \rangle + \langle u, \nabla^2_\theta F (\calZ_{-i}; \htheta^{(t-1)} )[\htheta^{(t-1)}_{-i} - \htheta^{(t-1)}] \rangle \\
	 & \quad + \langle u, (\nabla^2_\theta F (\calZ_{-i}; \theta' ) - \nabla^2_\theta F (\calZ_{-i}; \htheta^{(t-1)} ) ) [\htheta^{(t-1)}_{-i} - \htheta^{(t-1)}] \rangle
	\end{split}
	\end{equation} where $\theta' = a\htheta^{(t-1)} + (1-a)\htheta^{(t-1)}_{-i}$ for some $a\in[0,1]$.
	
	Plugging \eqref{eq: gradient-taylor-expansion} into \eqref{eq: quantity-1-main-bound}, we have for any $u \in \bbR^p$,
	\begin{equation} \label{eq: quantity-1-main-bound-2}
		\begin{split}
			\langle \htheta^{(t)} - \htheta^{(t)}_{-i}, u \rangle & = \langle \htheta^{(t-1)} - \htheta^{(t-1)}_{-i}  - \alpha_t \ell(Z_i;\htheta^{(t-1)} ) + \alpha_t  \nabla^2_\theta F(\calZ_{-i}; \htheta^{(t-1)}  )[\htheta^{(t-1)}_{-i}  - \htheta^{(t-1)}], u \rangle  \\
			& \quad  + \alpha_t \langle u, (\nabla^2_\theta F (\calZ_{-i}; \theta' ) - \nabla^2_\theta F (\calZ_{-i}; \htheta^{(t-1)} ) ) [\htheta^{(t-1)}_{-i} - \htheta^{(t-1)}] \rangle \\
			& = \left\langle \left(\I_p - \alpha_t   \nabla^2_\theta F(\calZ_{-i}; \htheta^{(t-1)}) \right)[ \htheta^{(t-1)} - \htheta^{(t-1)}_{-i} ] - \alpha_t \ell(Z_i;  \htheta^{(t-1)} ),u \right\rangle  \\
			& \quad + \alpha_t \langle u, (\nabla^2_\theta F (\calZ_{-i}; \theta' ) - \nabla^2_\theta F (\calZ_{-i}; \htheta^{(t-1)} ) ) [\htheta^{(t-1)}_{-i} - \htheta^{(t-1)}] \rangle.
		\end{split}
	\end{equation} 
	
	Since $\htheta^{(0)} = \htheta^{(0)}_{-i}$, the error bound of $\htheta^{(t)}$ to $\htheta^{(t)}_{-i}$ holds at $t = 0$. Suppose now $\| \htheta^{(t-1)} - \htheta^{(t-1)}_{-i} \|_2 \leq \frac{2 \eta_i}{(2 \lambda_0 - \gamma) n}$ holds for some $t \geq 1$, then by \eqref{eq: quantity-1-main-bound-2} and the fact $\| \htheta^{(t)} - \htheta^{(t)}_{-i} \|_2 = \sup_{u: \|u\|_2 = 1} \langle \htheta^{(t)} - \htheta^{(t)}_{-i},u \rangle $, we have
	\begin{equation}\label{ineq: iteration-dependent-bound1}
		\begin{split}
			\|\htheta^{(t)} - \htheta^{(t)}_{-i}\|_2 \overset{(a)}\leq & \left\| \I_p - \alpha_t   \nabla^2_\theta F(\calZ_{-i}; \htheta^{(t-1)}) \right\|_2 \|\htheta^{(t-1)} - \htheta^{(t-1)}_{-i}\|_2 + \alpha_t \eta_i + \alpha_t n \gamma  \|\htheta^{(t-1)} - \htheta^{(t-1)}_{-i}\|_2^2 \\
			\overset{(b)} \leq &  (1- \alpha_t n\lambda_0) \|\htheta^{(t-1)} - \htheta^{(t-1)}_{-i}\|_2 + \alpha_t \eta_i + \alpha_t n \gamma \|\htheta^{(t-1)} - \htheta^{(t-1)}_{-i}\|_2^2 \\
			\overset{(c)} \leq &  (1- \alpha_t n\lambda_0)  \frac{2 \eta_i}{(2 \lambda_0 - \gamma) n} + \alpha_t \eta_i +\alpha_t n \gamma \cdot\Big(\frac{2 \eta_i}{(2 \lambda_0 - \gamma) n}\Big)^2\\
			\overset{(d)} \leq & (1- \alpha_t n \lambda_0)  \frac{2 \eta_i}{(2 \lambda_0 - \gamma) n} + \alpha_t \eta_i + \alpha_t n \gamma \frac{ \eta_i}{(2\lambda_0 - \gamma) n} = \frac{2 \eta_i}{(2 \lambda_0 - \gamma) n},
		\end{split} 
	\end{equation} where (a) is by triangle inequality and Assumptions \ref{asm: gradient-control} and \ref{asm: third-derivative-control}; (b) is by Assumption \ref{asm: hessian-condition}; (c) is by the induction assumption; (d) is because $n \geq \frac{4 \|\eta\|_\infty }{2 \lambda_0 - \gamma}$. This proves the desired bound on $\| \htheta^{(t)} - \htheta^{(t)}_{-i} \|_2$.
	
	{\bf Step 2.} We now prove the bound on $\|\ttheta^{(t)}_{-i}-\htheta^{(t)}_{-i}\|_2$ by induction. First by the update rule of $ \htheta^{(t)}_{-i}$ and the Taylor expansion in \eqref{eq: gradient-taylor-expansion}, we have for any $u \in \bbR^p$,
	\begin{equation*}
		\begin{split}
			\langle \htheta^{(t)}_{-i}, u \rangle & =  \left \langle \htheta^{(t-1)}_{-i} - \alpha_t \left( \nabla_{\theta} F(\calZ_{-i}; \htheta^{(t-1)}  ) + \nabla^2_\theta F(\calZ_{-i}; \htheta^{(t-1)}  )[\htheta^{(t-1)}_{-i} -  \htheta^{(t-1)} ]   \right), u \right \rangle \\
			& \quad - \alpha_t \langle u, (\nabla^2_\theta F (\calZ_{-i}; \theta' ) - \nabla^2_\theta F (\calZ_{-i}; \htheta^{(t-1)} ) ) [\htheta^{(t-1)}_{-i} - \htheta^{(t-1)}] \rangle.
		\end{split}
	\end{equation*} Combining with \eqref{eqn: iter-IACV-GD}, we have
	\begin{equation} \label{eq: quantity-2-main-bound}
		\begin{split}
			\langle \htheta^{(t)}_{-i} - \ttheta_{-i}^{(t)}, u \rangle  &=  \left \langle \htheta^{(t-1)}_{-i} -  \ttheta_{-i}^{(t-1)} - \alpha_t  \nabla^2_\theta F(\calZ_{-i}; \htheta^{(t-1)}  ) [ \htheta^{(t-1)}_{-i} - \ttheta_{-i}^{(t-1)} ], u \right \rangle \\
			& \quad - \alpha_t \langle u, (\nabla^2_\theta F (\calZ_{-i}; \theta' ) - \nabla^2_\theta F (\calZ_{-i}; \htheta^{(t-1)} ) ) [\htheta^{(t-1)}_{-i} - \htheta^{(t-1)}] \rangle \\
			& = \left \langle \left(\I_p - \alpha_t  \nabla^2_\theta F(\calZ_{-i}; \htheta^{(t-1)})  \right) [ \htheta^{(t-1)}_{-i} - \ttheta_{-i}^{(t-1)} ], u \right \rangle \\
			& \quad - \alpha_t \langle u, (\nabla^2_\theta F (\calZ_{-i}; \theta' ) - \nabla^2_\theta F (\calZ_{-i}; \htheta^{(t-1)} ) ) [\htheta^{(t-1)}_{-i} - \htheta^{(t-1)}] \rangle.
		\end{split}
	\end{equation} 
	
	Since $\htheta^{(0)}_{-i} = \ttheta^{(0)}_{-i}$, the error bound holds at $t = 0$. Suppose now $\|\ttheta_{-i}^{(t-1)} - \htheta^{(t-1)}_{-i}\|_2 \leq \frac{4 \gamma \eta_i^2 }{\lambda_0 (2 \lambda_0 - \gamma)^2 n^2 }$ holds, then at time $t$, based on \eqref{eq: quantity-2-main-bound} and the fact that $\| \htheta^{(t)}_{-i} - \ttheta_{-i}^{(t)} \|_2 = \sup_{u:\|u\|_2=1}\langle \htheta^{(t)}_{-i} - \ttheta_{-i}^{(t)}, u\rangle$, we have
	\begin{equation}\label{ineq: iteration-dependent-bound2}
		\begin{split}
			\|\htheta^{(t)}_{-i} - \ttheta_{-i}^{(t)}\|_2 \overset{(a)}\leq & \left\| \I_p - \alpha_t  \nabla^2_\theta F(\calZ_{-i}; \htheta^{(t-1)}) \right\|_2 \|\htheta^{(t-1)}_{-i} - \ttheta_{-i}^{(t-1)} \|_2 + \alpha_t n \gamma \|\htheta^{(t-1)}_{-i}-\htheta^{(t-1)} \|_2^2 \\
			 \overset{(b)} \leq & (1 - \alpha_t n \lambda_0) \|\htheta^{(t-1)}_{-i} - \ttheta_{-i}^{(t-1)} \|_2 + \alpha_t n \gamma\|\htheta^{(t-1)}_{-i}-\htheta^{(t-1)} \|_2^2 \\
			 \overset{(c)} \leq & (1 - \alpha_t n \lambda_0) \frac{4 \gamma \eta_i^2 }{\lambda_0 (2 \lambda_0 - \gamma)^2 n^2 } + \alpha_t n \gamma \left( \frac{2 \eta_i}{(2 \lambda_0 - \gamma) n}  \right)^2 \\
			 = & \frac{4 \gamma \eta_i^2 }{\lambda_0 (2 \lambda_0 - \gamma)^2 n^2 }.
		\end{split}
	\end{equation} Here (a) is by the triangle inequality and Assumption \ref{asm: third-derivative-control}; (b) is by Assumption \ref{asm: hessian-condition}; (c) is by the induction assumption and the error bound we have proved for $\|\htheta^{(t-1)}_{-i}-\htheta^{(t-1)} \|_2$ in Step 1.
	This proves the desired bound on $\| \ttheta^{(t)}_{-i} - \htheta^{(t)}_{-i} \|_2$, and thus completes the proof of this theorem.
\end{proof}

\begin{proof}[Proof of Theorem \ref{th: model-assessment-bound}]
	First, for every $i \in [n]$, by the Taylor expansion, we have
	\begin{equation*}
		\begin{split}
			\ell(Z_i; \ttheta^{(t)}_{-i} ) = \ell(Z_i; \htheta^{(t)}_{-i}) +  \nabla_\theta \ell(Z_i;  \theta' )^\top [\ttheta^{(t)}_{-i}  -   \htheta^{(t)}_{-i} ] 
		\end{split}
	\end{equation*} where $\theta' = a \ttheta^{(t)}_{-i} + (1-a) \htheta^{(t)}_{-i}$ for some $0 \leq a \leq 1$. 
 	Thus
	\begin{equation} \label{eq: individual-error-bound}
		\begin{split}
			|\ell(Z_i; \ttheta^{(t)}_{-i}  ) - \ell(Z_i; \htheta^{(t)}_{-i})| \leq \|\nabla_\theta \ell(Z_i; \theta' ) \|_2\cdot \|\ttheta^{(t)}_{-i}  -   \htheta^{(t)}_{-i} \|_2 \leq \eta'_i   \|\ttheta^{(t)}_{-i}  -   \htheta^{(t)}_{-i} \|_2. 
		\end{split}
	\end{equation} Here the last inequality is by Assumption~\ref{asm: GD-CV-error}.
	This completes the proof of this theorem, once we plug in the bound on $\|\ttheta^{(t)}_{-i} - \htheta^{(t)}_{-i} \|_2$ obtained in Theorem \ref{th: GD-case-approx-error}.
\end{proof}

\begin{proof}[Proof of Theorem \ref{th: SGD-case-approx-error}] 
{\bf Step 1.} In this step, we prove the bound on $\|\htheta^{(t)} -\htheta^{(t)}_{-i}\|_2$ by induction. Similar to \eqref{eq: gradient-taylor-expansion}, given any $u \in \bbR^p$, by the Taylor expansion for $ \langle \nabla_{\theta} F(\calZ_{S_t\setminus \{ i\} }; \htheta^{(t-1)}_{-i}  ), u \rangle$, we have
\begin{equation} \label{eq: stochastic-gradient-taylor-expansion}
	\begin{split}
		\langle \nabla_{\theta} F(\calZ_{S_t\setminus \{ i\} }; \htheta^{(t-1)}_{-i}  ), u \rangle & = \left\langle \nabla_{\theta} F(\calZ_{S_t\setminus \{ i\} }; \htheta^{(t-1)}  ) + \nabla^2_\theta F(\calZ_{S_t\setminus \{ i\} }; \htheta^{(t-1)}  )[\htheta^{(t-1)}_{-i} -  \htheta^{(t-1)} ], u \right\rangle  \\
		& \quad  + \langle (\nabla^2_\theta F (\calZ_{S_t\setminus \{ i\}}; \theta' ) - \nabla^2_\theta F (\calZ_{S_t\setminus \{ i\}}; \htheta^{(t-1)} ) ) [\htheta^{(t-1)}_{-i} - \htheta^{(t-1)}], u \rangle,
	\end{split}
	\end{equation} where $\theta' = a \htheta^{(t-1)} + (1-a) \htheta^{(t-1)}_{-i}$ with $0 \leq a \leq 1$.
	
	So we have
	\begin{equation*}
		\begin{split}
			\langle \htheta^{(t)} - \htheta^{(t)}_{-i},u \rangle  & = \langle \htheta^{(t-1)} - \htheta^{(t-1)}_{-i} - \alpha_t\nabla_\theta F (\calZ_{S_t}; \htheta^{(t-1)} ) + \alpha_t\nabla_\theta F (\calZ_{S_t\setminus \{i \} }; \htheta^{(t-1)}_{-i} ), u \rangle  \\
			& \overset{ \eqref{eq: stochastic-gradient-taylor-expansion} } = \langle [\I_p - \alpha_t \nabla^2_\theta F (\calZ_{S_t\setminus \{i \} }; \htheta^{(t-1)} )  ][\htheta^{(t-1)} - \htheta^{(t-1)}_{-i}], u \rangle   \\
			& \quad- \langle \alpha_t\nabla_\theta F (\calZ_{S_t}; \htheta^{(t-1)} )- \alpha_t\nabla_\theta F (\calZ_{S_t\setminus \{i \} }; \htheta^{(t-1)} ), u \rangle  \\
			& \quad + \alpha_t  \langle (\nabla^2_\theta F (\calZ_{S_t\setminus \{ i\}}; \theta' ) - \nabla^2_\theta F (\calZ_{S_t\setminus \{ i\}}; \htheta^{(t-1)} ) ) [\htheta^{(t-1)}_{-i} - \htheta^{(t-1)}], u \rangle.
		\end{split}
	\end{equation*}
	
	Then we have
	\begin{equation} \label{ineq: sgd-main-bound}
		\begin{split}
			\bbE\|\htheta^{(t)} - \htheta^{(t)}_{-i} \|_2 & \leq \underbrace{\bbE\left( \|  [\I_p - \alpha_t \nabla^2_\theta F (\calZ_{S_t\setminus \{i \} }; \htheta^{(t-1)} )  ][\htheta^{(t-1)} - \htheta^{(t-1)}_{-i}]\|_2 \right)}_{\textnormal{(A)}} \\
			& \quad +\alpha_t  \underbrace{\bbE \| \nabla_\theta F (\calZ_{S_t}; \htheta^{(t-1)} ) - \nabla_\theta F (\calZ_{S_t\setminus \{i \} }; \htheta^{(t-1)} ) \|_2}_{\textnormal{(B)}} \\
			& \quad + \alpha_t  \underbrace{\bbE \|  (\nabla^2_\theta F (\calZ_{S_t\setminus \{ i\}}; \theta' ) - \nabla^2_\theta F (\calZ_{S_t\setminus \{ i\}}; \htheta^{(t-1)} ) ) [\htheta^{(t-1)}_{-i} - \htheta^{(t-1)}] \|_2}_{\textnormal{(C)}}.
		\end{split}
	\end{equation}
	
	Next, we bound the (A), (B), (C) terms separately. First for term (A),
	\begin{equation} \label{ineq: A-bound}
		\begin{split}
			& \bbE\left( \|   [\I_p - \alpha_t \nabla^2_\theta F (\calZ_{S_t\setminus \{i \} }; \htheta^{(t-1)} )  ][\htheta^{(t-1)} - \htheta^{(t-1)}_{-i}] \|_2  \right) \\
		= & \bbE_{S_1,\ldots, S_{t-1}}\left[ \bbE_{S_t} (\|  [\I_p - \alpha_t \nabla^2_\theta F (\calZ_{S_t\setminus \{i \} }; \htheta^{(t-1)} )  ][\htheta^{(t-1)} - \htheta^{(t-1)}_{-i}]  \|_2 )  \right] \\
			\overset{\textnormal{Assumption } \eqref{asm: sgd-hessian-condition} }\leq & (1 - \alpha_t K \lambda_0) \bbE_{S_1,\ldots, S_{t-1}} \left[ \|\htheta^{(t-1)} - \htheta^{(t-1)}_{-i} \|_2 \right].
		\end{split}
	\end{equation} Next, for term (B),
	\begin{equation} \label{ineq: B-bound}
		\begin{split}
			& \bbE \| \nabla_\theta F (\calZ_{S_t}; \htheta^{(t-1)} ) - \nabla_\theta F (\calZ_{S_t\setminus \{i \} }; \htheta^{(t-1)} ) \|_2 \\
			= &  \bbE_{S_1,\ldots, S_{t-1}}\left[ \bbE_{S_t}   \| \nabla_\theta F (\calZ_{S_t}; \htheta^{(t-1)} ) - \nabla_\theta F (\calZ_{S_t\setminus \{i \} }; \htheta^{(t-1)} ) \|_2  \right] \\
			= & \frac{K}{n} \bbE_{S_1,\ldots, S_{t-1}} \| \nabla_\theta\ell(Z_i; \htheta^{(t-1)} ) \|_2 \overset{ \textnormal{Assumption } \eqref{asm: sgd-gradient-control} }\leq \frac{ K \eta_i}{n}.
		\end{split}
	\end{equation} Here the next-to-last step holds because each data point is included independently in $S_t$ with probability $K/n$; thus with probability $1 - K/n$,  we have $\nabla_\theta F (\calZ_{S_t}; \htheta^{(t-1)} ) - \nabla_\theta F (\calZ_{S_t\setminus \{i \} }; \htheta^{(t-1)} ) = 0 $ and with probability $K/n$, $i \in S_t$ and in that case we have $\nabla_\theta F (\calZ_{S_t}; \htheta^{(t-1)} ) - \nabla_\theta F (\calZ_{S_t\setminus \{i \} }; \htheta^{(t-1)} ) =  \nabla_\theta \ell(Z_i; \htheta^{(t-1)} ) $.
	
	And, for term (C),
	\begin{equation} \label{ineq: C-bound}
		\begin{split}
			&\bbE \|  (\nabla^2_\theta F (\calZ_{S_t\setminus \{ i\}}; \theta' ) - \nabla^2_\theta F (\calZ_{S_t\setminus \{ i\}}; \htheta^{(t-1)} ) ) [\htheta^{(t-1)}_{-i} - \htheta^{(t-1)}] \|_2 \\
			\leq &\bbE \sup_{b\in[0,1]}\|  (\nabla^2_\theta F (\calZ_{S_t\setminus \{ i\}};  b \htheta^{(t-1)} + (1-b) \htheta^{(t-1)}_{-i} ) - \nabla^2_\theta F (\calZ_{S_t\setminus \{ i\}}; \htheta^{(t-1)} ) ) [\htheta^{(t-1)}_{-i} - \htheta^{(t-1)}] \|_2 \\
			\overset{\textnormal{Assumption }\ref{asm: sgd-third-derivative-control} } \leq & \gamma K\bbE_{S_1,\ldots, S_{t-1}}\left[ \|\htheta^{(t-1)}_{-i} - \htheta^{(t-1)}\|_2^2  \right] \\
			\overset{\textnormal{Assumption }\ref{asm: reverse-jensen} } \leq & \gamma K \beta \frac{n}{K} \left( \bbE  \|\htheta^{(t-1)}_{-i} - \htheta^{(t-1)}\|_2 \right)^2 = \gamma \beta n \left( \bbE  \|\htheta^{(t-1)}_{-i} - \htheta^{(t-1)}\|_2 \right)^2.
		\end{split}
	\end{equation}
	
	Since $\htheta^{(0)} = \htheta^{(0)}_{-i}$, the error bound of $\htheta^{(t)}$ to $\htheta^{(t)}_{-i}$ holds at $t = 0$. Suppose now $\bbE \| \htheta^{(t-1)} - \htheta^{(t-1)}_{-i} \|_2 \leq  \frac{2 \eta_i}{(2 \lambda_0 - \gamma\beta) n}$ holds for some $t \geq 1$, then by plugging \eqref{ineq: A-bound}, \eqref{ineq: B-bound} and \eqref{ineq: C-bound} into \eqref{ineq: sgd-main-bound}, we have 
	\begin{equation} \label{ineq: sgd-approx-bound}
		\begin{split}
			\bbE\|\htheta^{(t)} - \htheta^{(t)}_{-i} \|_2 & \leq  (1 - \alpha_t K \lambda_0) \bbE \|\htheta^{(t-1)} - \htheta^{(t-1)}_{-i} \|_2 + \frac{\alpha_t K \eta_i}{n} + \alpha_t \gamma \beta n \left( \bbE  \|\htheta^{(t-1)}_{-i} - \htheta^{(t-1)}\|_2 \right)^2 \\
			& \leq (1 - \alpha_t K \lambda_0)   \frac{2 \eta_i}{(2 \lambda_0 - \gamma\beta) n} + \frac{\alpha_t K \eta_i}{n} + \alpha_t K \gamma \beta \frac{ \eta_i}{(2 \lambda_0 - \gamma\beta) n}\frac{4 \eta_i}{(2 \lambda_0 - \gamma\beta) K}\\
			& \leq  (1 - \alpha_t K \lambda_0)   \frac{2 \eta_i}{(2 \lambda_0 - \gamma\beta) n} + \frac{\alpha_t K \eta_i}{n} + \alpha_t K \gamma \beta \frac{ \eta_i}{(2 \lambda_0 - \gamma\beta) n}  =  \frac{2 \eta_i}{(2 \lambda_0 - \gamma\beta) n},
		\end{split} 
	\end{equation} where the next-to-last step holds since we have assumed $K \geq \frac{4 \|\eta\|_\infty }{2 \lambda_0 - \gamma \beta}$ by the assumption on $K$. This finishes the proof for the first part.

{\bf Step 2.} We now prove the bound on $\|\ttheta_{-i}^{(t)} -\htheta^{(t)}_{-i}\|_2$ by induction. Based on the Taylor expansion in \eqref{eq: stochastic-gradient-taylor-expansion} and the update rule of $\ttheta_{-i}^{(t)}$ in \eqref{eqn: iter-IACV-SGD}, we have for any $u \in \bbR^p$,
\begin{equation*}
	\begin{split}
		\langle \htheta^{(t)}_{-i} - \ttheta_{-i}^{(t)}, u \rangle  &= \langle \htheta^{(t-1)}_{-i} -  \alpha_t\nabla_\theta F (\calZ_{S_t\setminus \{i \} }; \htheta^{(t-1)}_{-i} ) -\ttheta_{-i}^{(t-1)}, u \rangle  \\
		& \quad + \left\langle \alpha_t \left( \nabla_\theta F(\calZ_{S_t \setminus \{i \} }; \htheta^{(t-1)}) + \nabla^2_\theta F(\calZ_{S_t \setminus \{i \} }; \htheta^{(t-1)})[ \ttheta_{-i}^{(t-1)} - \htheta^{(t-1)} ] \right), u \right\rangle  \\
		& \overset{ \eqref{eq: stochastic-gradient-taylor-expansion} }= \left\langle \left(\I_p - \alpha_t  \nabla^2_\theta F(\calZ_{S_t \setminus \{i \}}; \htheta^{(t-1)})  \right) [ \htheta^{(t-1)}_{-i} - \ttheta_{-i}^{(t-1)} ], u \right \rangle \\
			& \quad - \alpha_t \langle (\nabla^2_\theta F (\calZ_{S_t\setminus \{ i\}}; \theta' ) - \nabla^2_\theta F (\calZ_{S_t\setminus \{ i\}}; \htheta^{(t-1)} ) ) [\htheta^{(t-1)}_{-i} - \htheta^{(t-1)}], u \rangle.
	\end{split}
\end{equation*}

Thus we have
\begin{equation} \label{ineq: sgd-main-bound2}
	\begin{split}
		\bbE \| \htheta^{(t)}_{-i} - \ttheta_{-i}^{(t)} \|_2 & \leq \bbE\left[ \left( \| \I_p - \alpha_t  \nabla^2_\theta F(\calZ_{S_t \setminus \{i \}}; \htheta^{(t-1)})  \right) [ \htheta^{(t-1)}_{-i} - \ttheta_{-i}^{(t-1)} ] \|_2 \right] \\
		& \quad + \alpha_t \bbE \left( \|(\nabla^2_\theta F (\calZ_{S_t\setminus \{ i\}}; \theta' ) - \nabla^2_\theta F (\calZ_{S_t\setminus \{ i\}}; \htheta^{(t-1)} ) ) [\htheta^{(t-1)}_{-i} - \htheta^{(t-1)}] \|^2_2\right)\\
		& \quad \leq (1- \alpha_t K \lambda_0 ) \bbE \|\htheta^{(t-1)}_{-i} - \ttheta_{-i}^{(t-1)} \|_2 + \alpha_t \gamma \beta n \left( \bbE \|\htheta^{(t-1)}_{-i}-\htheta^{(t-1)} \|_2  \right)^2 \\
		& \leq  (1- \alpha_t K \lambda_0 ) \bbE \|\htheta^{(t-1)}_{-i} - \ttheta_{-i}^{(t-1)} \|_2 + \alpha_t \gamma \beta n   \left( \frac{2 \eta_i}{(2 \lambda_0 - \gamma\beta) n} \right)^2.
	\end{split}
\end{equation} Here the next-to-last step holds by a similar arguments as in \eqref{ineq: A-bound} and \eqref{ineq: C-bound}, while the last step holds by the error bound we have proved for $\bbE \|\htheta^{(t-1)}_{-i}-\htheta^{(t-1)} \|_2 $ in Step 1.

Since $\htheta^{(0)}_{-i} = \ttheta^{(0)}_{-i}$, the error bound holds at $t = 0$. Suppose now $\bbE \|\ttheta_{-i}^{(t-1)} - \htheta^{(t-1)}_{-i}\|_2 \leq\frac{4 \gamma \beta \eta_i^2 }{\lambda_0 (2 \lambda_0 - \gamma\beta)^2 nK }$ holds, then at time $t$, based on \eqref{ineq: sgd-main-bound2}, we have
\begin{equation} \label{ineq: sgd-cv-bound}
	\begin{split}
		\bbE \| \htheta^{(t)}_{-i} - \ttheta_{-i}^{(t)} \|_2 & \leq (1- \alpha_t K \lambda_0 )  \frac{4 \gamma \beta \eta_i^2 }{\lambda_0 (2 \lambda_0 - \gamma\beta)^2 nK } + \alpha_t \gamma n \beta  \left( \frac{2 \eta_i}{(2 \lambda_0 - \gamma\beta) n} \right)^2 =  \frac{4 \gamma \beta \eta_i^2 }{\lambda_0 (2 \lambda_0 - \gamma\beta)^2 n K },
	\end{split}
\end{equation} which proves the desired bound and thus completes the proof of this theorem.

\end{proof}

\begin{proof}[Proof of Theorem \ref{th: model-assessment-bound-SGD}]
	First, for every $i \in [n]$, by the Taylor expansion, we have
	\begin{equation*}
		\begin{split}
			\ell(Z_i; \ttheta^{(t)}_{-i} ) = \ell(Z_i; \htheta^{(t)}_{-i}) + \nabla_\theta \ell(Z_i;  \theta' )^\top [ \ttheta^{(t)}_{-i}  -   \htheta^{(t)}_{-i} ],
		\end{split}
	\end{equation*} where $\theta' = a \ttheta^{(t)}_{-i} + (1-a) \htheta^{(t)}_{-i}$ for some $0 \leq a \leq 1$. 
 
 Next 
	\begin{equation*}
		\begin{split}
			\bbE(\textnormal{Err}_{\textnormal{CV}}^{(t)})  &\leq \frac{1}{n} \sum_{i=1}^n \bbE(| \ell(Z_i; \ttheta^{(t)}_{-i}  ) - \ell(Z_i; \htheta^{(t)}_{-i}) |) \\
			& \leq \frac{1}{n} \sum_{i=1}^n \bbE \left\{\sup_{a \in [0,1]} | \nabla_\theta \ell(Z_i;  a \ttheta^{(t)}_{-i} + (1-a) \htheta^{(t)}_{-i} )^\top [ \ttheta^{(t)}_{-i}  -   \htheta^{(t)}_{-i} ] | \right\} \\
    &\leq \frac{1}{n} \sum_{i=1}^n \frac{4 \eta'_i  \eta_i^2\gamma \beta }{\lambda_0 (2 \lambda_0 - \gamma\beta)^2 nK }, 
		\end{split}
	\end{equation*}
where	the last inequality is by Assumption \ref{asm:SGD-CV-error} and Theorem \ref{th: SGD-case-approx-error}.
\end{proof}

\begin{proof}[Proof of Theorem \ref{th: Prox-GD-case-approx-error}]
The proof strategy of this theorem is similar to the proof of Theorem \ref{th: GD-case-approx-error}.
For convenience, we will define
\begin{equation} \label{def: prox-operator}
    \prox_{h}^{\alpha_t}(\theta'):= \arg\!\min_\theta\left\{\frac{1}{2\alpha_t}\|\theta - \theta'\|^2_2 + h(\theta)\right\}.
\end{equation}
A key property we use here is that the proximal operator $\prox_{h}^{\alpha_t}$ is nonexpansive for convex $h$ \citep[Proposition 8.19]{wright2022optimization}, i.e.,
\[\|\prox_{h}^{\alpha_t}(x) - \prox_{h}^{\alpha_t}(x')\|_2 \leq \|x-x'\|_2.\]

	{\bf Step 1.} In this step, we derive the error bound for $\| \htheta^{(t)} - \htheta^{(t)}_{-i} \|_2$. By the updating rules of $\htheta^{(t)}$ and $\htheta^{(t)}_{-i}$, we have 
	\begin{equation*}
		\begin{split}
			\| \htheta^{(t)} - \htheta^{(t)}_{-i} \|_2 &= \left\|   \prox_{h}^{\alpha_t} \left( \htheta^{(t-1)} - \alpha_t \nabla_{\theta} g(\calZ; \htheta^{(t-1)} ) \right) -  \prox_{h}^{\alpha_t} \left( \htheta^{(t-1)}_{-i} - \alpha_t \nabla_{\theta} g(\calZ_{-i}; \htheta^{(t-1)}_{-i} ) \right)  \right\|_2 \\
			&  \leq \left\| \htheta^{(t-1)} - \alpha_t \nabla_{\theta} g(\calZ; \htheta^{(t-1)} ) - (\htheta^{(t-1)}_{-i} - \alpha_t \nabla_{\theta} g(\calZ_{-i}; \htheta^{(t-1)}_{-i} )) \right \|_2.
		\end{split}
	\end{equation*} Here the inequality holds by the nonexpansiveness property of the proximal operator.
	The rest of the proof follows exactly as  in Step 1 of Theorem \ref{th: GD-case-approx-error} by replacing $F(\cdot; \theta)$ with $g(\cdot;\theta)$.
	
	{\bf Step 2.} By our estimator in \eqref{eqn: iter-IACV-ProxGD}, we have
	\begin{equation*}
		\begin{split}
			\| \ttheta_{-i}^{(t)} - \htheta^{(t)}_{-i} \|_2 &= \Big\| \prox_{h}^{\alpha_t}\left(\ttheta_{-j}^{(t-1)} - \alpha_t ( \nabla_\theta g(\calZ_{-j}; \htheta^{(t-1)} ) + \nabla^2_\theta g( \calZ_{-j}; \htheta^{(t-1)} )[ \ttheta_{-j}^{(t-1)} - \htheta^{(t-1)} ] ) \right)  \\
			& \quad -   \prox_{h}^{\alpha_t} \left( \htheta^{(t-1)}_{-i} - \alpha_t \nabla_{\theta} g(\calZ_{-i}; \htheta^{(t-1)}_{-i} ) \right) \Big\|_2 \\
			& \leq \| \ttheta_{-j}^{(t-1)} - \alpha_t ( \nabla_\theta g(\calZ_{-j}; \htheta^{(t-1)} ) + \nabla^2_\theta g( \calZ_{-j}; \htheta^{(t-1)} )[ \ttheta_{-j}^{(t-1)} - \htheta^{(t-1)} ] ) - \htheta^{(t-1)}_{-i} \\
			& \quad + \alpha_t \nabla_{\theta} g(\calZ_{-i}; \htheta^{(t-1)}_{-i} ) \|_2,
		\end{split}
	\end{equation*} where the inequality is again due to the nonexpansiveness property of the proximal operator.
	The remainder of the proof again follows the same argument as in Step 2 of Theorem \ref{th: GD-case-approx-error} by replacing $F(\cdot; \theta)$ with $g(\cdot;\theta)$.
\end{proof}

\subsection{Proofs of convergence results} \label{proof-sec: limiting-behavior}
In this section we prove Theorems~\ref{th: limiting-behavior-estimator-GD} and~\ref{th: limiting-behavior-estimator-prox-GD}, which establish convergence of the IACV estimator to the NS estimator in the GD and ProxGD setting, respectively.

\begin{proof}[Proof of Theorem \ref{th: limiting-behavior-estimator-GD}]
Since we assume $\htheta^{(t)}$ converges to $\htheta$ and Assumption \ref{asm: hessian-condition} is satisfied at all iterations, we therefore see that $\nabla^2_\theta F ( \calZ_{-i}; \htheta )$ is invertible by the continuity of $\nabla^2_\theta F ( \calZ_{-i}; \theta )$.  
	Let us denote $$b_i^t = \htheta^{(t-1)} - (\nabla^2_{\theta} F ( \calZ_{-i}; \htheta^{(t-1)} ) )^{-1} \nabla_{\theta} F( \calZ_{-i}; \htheta^{(t-1)} ).$$
	
	First, based on the update rule for $\ttheta_{-i}^{(t)}$ in \eqref{eqn: iter-IACV-GD}, we have:
	\begin{equation*} \label{eq: convergence-update-equation}
		\begin{split}
			\ttheta_{-i}^{(t)} - b_i^t = [\I_p - \alpha \nabla^2_{\theta} F ( \calZ_{-i}; \htheta^{(t-1)} ) ] (\ttheta_{-i}^{(t-1)} - b_i^t ).
		\end{split}
	\end{equation*}
	This implies
	\begin{equation*}
		\begin{split}
			\ttheta_{-i}^{(t)} - \ttheta^{\NS}_{-i}  &=  [\I_p - \alpha \nabla^2_{\theta} F ( \calZ_{-i}; \htheta^{(t-1)} ) ] (\ttheta_{-i}^{(t-1)} - b_i^t ) + b_i^t - \ttheta^{\NS}_{-i}\\
			& =  [\I_p - \alpha \nabla^2_{\theta} F ( \calZ_{-i}; \htheta^{(t-1)} ) ](\ttheta_{-i}^{(t-1)} - \ttheta^{\NS}_{-i}) -  \alpha\nabla^2_{\theta} F ( \calZ_{-i}; \htheta^{(t-1)} )(\ttheta^{\NS}_{-i} - b_i^t),
		\end{split}
	\end{equation*} 
	and thus
 \begin{equation*}
		\begin{split}\|\ttheta_{-i}^{(t)} - \ttheta^{\NS}_{-i}\|_2
 &\leq \|\I_p - \alpha \nabla^2_{\theta} F ( \calZ_{-i}; \htheta^{(t-1)} )\|\cdot \|\ttheta_{-i}^{(t-1)} - \ttheta^{\NS}_{-i}\|_2 + \alpha \|\nabla^2_{\theta} F ( \calZ_{-i}; \htheta^{(t-1)} )\|\cdot \| \ttheta^{\NS}_{-i} - b_i^t\|_2\\
 &\leq  (1-\alpha n\lambda_0)\cdot \|\ttheta_{-i}^{(t-1)} - \ttheta^{\NS}_{-i}\|_2 +  \| \ttheta^{\NS}_{-i} - b_i^t\|_2,\end{split}\end{equation*}
 where the last step holds by Assumption~\ref{asm: hessian-condition} together with the choice of $\alpha$.

	Note that since $\nabla^2_\theta F(\calZ_{-i}; \theta) $ and $ \nabla_\theta F (\calZ_{-i}; \theta)$ are continuous in $\theta$, we then have $b_i^t \to \ttheta^{\NS}_{-i}$ as $\htheta^{(t)}$ converges to $\htheta$, i.e.,
 \[\lim_{t\rightarrow\infty} \|\ttheta^{\NS}_{-i}-b^t_i\|_2 = 0.\]
 Combined with the above, this is sufficient to prove that 
 \[\lim_{t\rightarrow \infty}\|\ttheta_{-i}^{(t)} - \ttheta^{\NS}_{-i}\|_2 = 0,\]
 which completes the proof of the theorem. \end{proof}

\begin{proof}[Proof of Theorem \ref{th: limiting-behavior-estimator-prox-GD}]
	Let us first prove the claim that $\ttheta^{\NS}_{-i}$ in \eqref{eq: one-step-prox-Newton-estimator} is a fixed point of 
\begin{equation} \label{eq: proximal-gd-fixed-point-equation}
	a^{(t)} =\prox_{h}^{\alpha} \left(  a^{(t-1)} - \alpha \left( \nabla_\theta g(\calZ_{-i}; \htheta ) + \nabla^2_\theta g(\calZ_{-i}; \htheta ) [a^{(t-1)} - \htheta  ]   \right) \right),
\end{equation} where $\prox$ is defined in \eqref{def: prox-operator}.
To show this, it is equivalent to show $$\ttheta^{\NS}_{-i} = \prox_{h}^{\alpha} \left( \ttheta^{\NS}_{-i} - \alpha \left( \nabla_\theta g(\calZ_{-i}; \htheta ) + \nabla^2_\theta g(\calZ_{-i}; \htheta ) [\ttheta^{\NS}_{-i} - \htheta  ]   \right) \right). $$
	
	Since $h$ is convex, by the definition of $\prox$ in \eqref{eq: prox-bar-definition} and the definition of $\ttheta^{\NS}_{-i}$ in \eqref{eq: one-step-prox-Newton-estimator}, the first-order optimality condition tells us that $\ttheta^{\NS}_{-i}$ must satisfy
	\begin{equation} \label{eq: first-order-optimality-condition}
		 0 \in - \left(\nabla^2_\theta g(\calZ_{-i}; \htheta) \left(   \htheta - \ttheta^{\NS}_{-i} \right) - \nabla_\theta g(\calZ_{-i}; \htheta )\right) +  \partial h(\ttheta^{\NS}_{-i}).
	\end{equation}
	
 Also, the proximal operator $ \prox_{h}^{\alpha}$ has a unique solution as $h$ is convex. By the definition of the proximal operator $\prox_{h}^{\alpha }$, the output of $\prox_{h}^{\alpha} \left( \ttheta^{\NS}_{-i} - \alpha \left( \nabla_\theta g(\calZ_{-i}; \htheta ) + \nabla^2_\theta g(\calZ_{-i}; \htheta ) [\ttheta^{\NS}_{-i}  - \htheta  ]   \right) \right)$, say $x^*$, is the unique vector that satisfies
	\begin{equation} \label{eq: first-order-opt-condition2}
		\begin{split}
			0 \in -\frac{\ttheta^{\NS}_{-i} - \alpha \left( \nabla_\theta g(\calZ_{-i}; \htheta ) + \nabla^2_\theta g(\calZ_{-i}; \htheta ) [\ttheta^{\NS}_{-i} - \htheta  ]  \right) - x^*}{\alpha} + \partial h(x^*).
		\end{split}
	\end{equation} Notice that $x^* = \ttheta^{\NS}_{-i}$ satisfies \eqref{eq: first-order-opt-condition2} because of \eqref{eq: first-order-optimality-condition}. Thus, $\ttheta^{\NS}_{-i}$ is a fixed point of~\eqref{eq: proximal-gd-fixed-point-equation} as desired.

Next, we show $\ttheta^{(t)}_{-i}$ in \eqref{eqn: iter-IACV-ProxGD} converges to $\ttheta^{\NS}_{-i}$ under the assumptions given in the statement of the theorem. By \eqref{eqn: iter-IACV-ProxGD} together with the calculations above, we have
\begin{equation*}
	\begin{split}
		\ttheta^{(t)}_{-i} - \ttheta^{\NS}_{-i} &=  \prox_{h}^{\alpha}\left(\ttheta_{-i}^{(t-1)} - \alpha ( \nabla_\theta g(\calZ_{-i}; \htheta^{(t-1)} ) + \nabla^2_\theta g( \calZ_{-i}; \htheta^{(t-1)} )[ \ttheta_{-i}^{(t-1)} - \htheta^{(t-1)} ] ) \right) \\
		&\quad  - \prox_{h}^{\alpha} \left( \ttheta^{\NS}_{-i} - \alpha \left( \nabla_\theta g(\calZ_{-i}; \htheta ) + \nabla^2_\theta g(\calZ_{-i}; \htheta ) [\ttheta^{\NS}_{-i} - \htheta  ]   \right) \right).
	\end{split}
\end{equation*} So we have
\begin{equation} \label{eq: proximal-case-convergence-iteration-bound}
	\begin{split}
		\|\ttheta^{(t)}_{-i} - \ttheta^{\NS}_{-i}\|_2 &= \Big\|   \prox_{h}^{\alpha}\left(\ttheta_{-i}^{(t-1)} - \alpha ( \nabla_\theta g(\calZ_{-i}; \htheta^{(t-1)} ) + \nabla^2_\theta g( \calZ_{-i}; \htheta^{(t-1)} )[ \ttheta_{-i}^{(t-1)} - \htheta^{(t-1)} ] ) \right) \\
		& \quad -  \prox_{h}^{\alpha} \left( \ttheta^{\NS}_{-i} - \alpha \left( \nabla_\theta g(\calZ_{-i}; \htheta ) + \nabla^2_\theta g(\calZ_{-i}; \htheta ) [\ttheta^{\NS}_{-i} - \htheta  ]   \right) \right) \Big\|_2 \\
		& \leq \Big\|\left(\ttheta_{-i}^{(t-1)} - \alpha ( \nabla_\theta g(\calZ_{-i}; \htheta^{(t-1)} ) + \nabla^2_\theta g( \calZ_{-i}; \htheta^{(t-1)} )[ \ttheta_{-i}^{(t-1)} - \htheta^{(t-1)} ] ) \right) \\
		& \quad  -  \left( \ttheta^{\NS}_{-i} - \alpha \left( \nabla_\theta g(\calZ_{-i}; \htheta ) + \nabla^2_\theta g(\calZ_{-i}; \htheta ) [\ttheta^{\NS}_{-i} - \htheta  ]   \right) \right) \Big\|_2 \\
  & = \Big\|  [\I - \alpha \nabla^2_\theta g( \calZ_{-i}; \htheta^{(t-1)} ) ](  \ttheta^{(t-1)}_{-i} - \ttheta^{\NS}_{-i} )  + \alpha ( \nabla^2_\theta g(\calZ_{-i}; \htheta ) - \nabla^2_\theta g( \calZ_{-i}; \htheta^{(t-1)} ) ) \ttheta^{\NS}_{-i} \\
		& \quad  - \alpha( \nabla_\theta g(\calZ_{-i}; \htheta^{(t-1)} ) - \nabla^2_\theta g( \calZ_{-i}; \htheta^{(t-1)} ) \htheta^{(t-1)} ) + \alpha ( \nabla_\theta g(\calZ_{-i}; \htheta )- \nabla^2_\theta g(\calZ_{-i}; \htheta ) \htheta ) \Big\|_2 \\
		& \leq \left\|  [\I - \alpha \nabla^2_\theta g( \calZ_{-i}; \htheta^{(t-1)} ) ](  \ttheta^{(t-1)}_{-i} - \ttheta^{\NS}_{-i} ) \right\|_2 + \alpha\Big[ \| ( \nabla^2_\theta g(\calZ_{-i}; \htheta ) - \nabla^2_\theta g( \calZ_{-i}; \htheta^{(t-1)} ) ) \ttheta^{\NS}_{-i}  \|_2 \\
		& \quad + \| \nabla_\theta g(\calZ_{-i}; \htheta^{(t-1)} ) - \nabla_\theta g(\calZ_{-i}; \htheta )  \|_2 +  \| \nabla^2_\theta g( \calZ_{-i}; \htheta^{(t-1)} ) \htheta^{(t-1)} - \nabla^2_\theta g(\calZ_{-i}; \htheta ) \htheta  \|_2\Big]\\
  &\leq (1-\alpha n\lambda_0)\cdot  \|  \ttheta^{(t-1)}_{-i} - \ttheta^{\NS}_{-i} \|_2 + \alpha\Big[ \| ( \nabla^2_\theta g(\calZ_{-i}; \htheta ) - \nabla^2_\theta g( \calZ_{-i}; \htheta^{(t-1)} ) ) \ttheta^{\NS}_{-i}  \|_2 \\
		& \quad + \| \nabla_\theta g(\calZ_{-i}; \htheta^{(t-1)} ) - \nabla_\theta g(\calZ_{-i}; \htheta )  \|_2 +  \| \nabla^2_\theta g( \calZ_{-i}; \htheta^{(t-1)} ) \htheta^{(t-1)} - \nabla^2_\theta g(\calZ_{-i}; \htheta ) \htheta  \|_2\Big].
	\end{split}
\end{equation} where the second step holds  because $\prox_{h}^{\alpha}$ is nonexpansive for convex $h$ \citep[Proposition 8.19]{wright2022optimization}, and the last step applies Assumption~\ref{asm: hessian-condition} (with $g$ in place of $F$, as assumed in the Theorem).
Since $\nabla^2_\theta g(\calZ_{-i}; \theta)$, $\nabla_\theta g(\calZ_{-i}; \theta)$ are continuous in $\theta$, and $\htheta^{(t)}$ is assumed to converge to $\htheta$, 
the quantity in square brackets converges to zero.
As in the proof of Theorem~\ref{th: limiting-behavior-estimator-GD}, this therefore implies
\[\lim_{t\rightarrow \infty} \|\ttheta^{(t)}_{-i} - \ttheta^{\NS}_{-i}\|_2\rightarrow 0.\]	 \end{proof}

\paragraph{Can Theorem~\ref{th: limiting-behavior-estimator-GD} be extended to SGD?}
	A key intuition behind the proof of Theorem \ref{th: limiting-behavior-estimator-GD} is that $\ttheta^{\NS}_{-i}$ in \eqref{eqn:NS-intro} is a fixed point of the following update equation
\begin{equation*}
		x^{t} =  x^{t-1} - \alpha_t \Big( \nabla_\theta F(\calZ_{-i}; \htheta ) 
		 + \nabla^2_\theta F(\calZ_{-i}; \htheta ) [x^{t-1} - \htheta  ]   \Big).
\end{equation*}	In addition, we use the fact that given any $\epsilon >0$, a sequence of scalars $\{ x^t \}$ that satisfies $x^t = (1 - \alpha_t) x^{t-1} + \epsilon$ must have $x^t\rightarrow 0$ when we take $\alpha_t \equiv \alpha > 0$. 

 In the SGD setting, on the other hand, the techniques for proving Theorem \ref{th: limiting-behavior-estimator-GD} fail as $\alpha_t$ there are required to decay to zero eventually for convergence. Thus, if we have a sequence of values $\{x^t\}$ satisfying $x^t = (1 - \alpha_t) x^{t-1} + \epsilon$  for any $\epsilon>0$, this does not necessarily imply $x^t\rightarrow 0$. Whether a result analogous to Theorem~\ref{th: limiting-behavior-estimator-GD} might be provable for SGD via a different technique remains an open question.

\end{document}